\documentclass[12pt,reqno]{amsart}
\usepackage[utf8]{inputenc}
\usepackage[T1]{fontenc}
\usepackage{amsmath}
\usepackage{amsthm}
\usepackage{amssymb}
\usepackage{dsfont}
\usepackage{graphicx}
\usepackage{hyperref}
\usepackage[font=small,labelfont=bf,tableposition=top]{caption}
\usepackage{cleveref}
\usepackage{pgfplots}
\usepackage{booktabs}
\DeclareCaptionLabelFormat{andtable}{#1~#2  \&  \tablename~\thetable}

\Crefname{convention}{Convention}{Conventions}
\DeclareCaptionLabelFormat{andtable}{#1~#2  \&  \tablename~\thetable}
\newtheorem{theorem}{Theorem}
\newtheorem{lemma}[theorem]{Lemma}
\newtheorem{proposition}[theorem]{Proposition}

\theoremstyle{remark}
\newtheorem{remark}{Remark}
\newtheorem{convention}{Convention}

\Crefname{problem}{Problem}{Problems}
\newcommand\pcref[1]{{\sf(\Cref{#1})}}
\newcommand\R{{\ensuremath {\mathbb R} }}

\newcommand\G{{\ensuremath {\mathbb G} }}
\newcommand\N{{\ensuremath {\mathbb N} }}
\newcommand\Pro{{\ensuremath {\mathbb P} }}
\newcommand\cM{{\ensuremath {\mathcal M} }}
\newcommand\cF{{\ensuremath {\mathcal F} }}
\newcommand\cE{{\ensuremath {\mathcal E} }}

\newcommand\cR{{\ensuremath {\mathcal R} }}
\newcommand\1{{\ensuremath {\mathds 1} }}
\newcommand\intRd{{\ensuremath {\int_{\mathbb{R}^d}}}}
\newcommand{\eps}{\epsilon}
\renewcommand{\epsilon}{\varepsilon}

\newcommand{\supp}{{\rm supp}}

\renewcommand{\geq}{\geqslant}
\renewcommand{\leq}{\leqslant}
\newcommand{\rhoext}{{\rho_{\rm ext}}}
\newcommand{\rhoextbeta}{{\rho_{{\rm ext}, \beta}}}
\newcommand{\muext}{{\mu_{\rm ext}}}
\renewcommand{\d}{{\rm d}}
\renewcommand{\r}{\mathbf{r}}
\newcommand{\wsl}{\overset{\ast}{\rightharpoonup}}

\newcommand{\e}[1]{\emph{#1}}
\renewcommand{\b}[1]{\textbf{#1}}
\newcommand{\ext}{{\rm ext}}
\newcommand{\SCE}{{\rm SCE}}
\begin{document}
\title[Dual charge approach to the MOT with Coulomb cost]{An external dual charge approach to the Optimal Transport with Coulomb cost}

\author[R. Lelotte]{Rodrigue LELOTTE}
\address{CEREMADE, Universit\'e Paris-Dauphine, PSL Research University, Place de Lattre de Tassigny, 75016 Paris, France.}
\email{lelotte@ceremade.dauphine.fr}

\date{\today}

\begin{abstract}
In this paper, we study the multimarginal optimal transport with Coulomb cost, also known in the physics literature as the \e{Strictly-Correlated Electrons} (SCE) functional. We prove that the dual Kantorovich potential is an electrostatic potential induced by an external charge density, which we call the \e{dual charge}. We study its properties and use it to discretize the potential in one and three space dimensions. 
\end{abstract}

\maketitle
{\hypersetup{linkcolor=black}
\tableofcontents
}
\section{Introduction}\label{sec:intro}

In the recent years, \e{Multimarginal Optimal Transport} (MOT) has begun to attract considerable attention, due to a wide variety of emerging applications outside of mathematics, such as economics, finance, physics and image processing (see \cite{pass2014multimarginal} for a rather detailed review and citations therein). As such, it has become of valuable importance to develop numerical methods to solve this problem, which is plagued with the infamous \e{curse of dimensionality}.

In physics, MOT appears in a variety of applications, \e{e.g.} in defining the \e{Uniform Electron Gas} (UEG; see \cite{lewin2018statistical}), which in turn serves as a building block for the very important \e{Local Density Approximation} in \e{Density Functional Theory} (DFT; see \cite{lewin2019universal}), a successful computational modeling method in quantum physics. Another example of MOT in physics, and closely related to DFT, is the paradigm of \e{classical} DFT (see \e{e.g.} \cite{yang1976molecular}), which exactly reformulates as a MOT problem together with \e{entropic regularization}. 

From a numerical viewpoint, MOT is notoriously hard to solve. Most algorithms require exponential time in the number of marginals and in the discretization size of their supports. While for some specific costs, it is possible to fashion polynomial time methods, in the case of the Coulomb cost the problem is known to be $\mathcal{NP}$-hard  \cite{altschuler2021hardness,altschuler2020polynomial}. Current methods then consist in either adding an entropic regularization \cite{benamou2016numerical}, which physically translates into adding positive temperature to the system, and which makes the numerics more tractable by heavily exploiting parallel computing; or using ideas from \e{Linear Programming} (\e{e.g.} leveraging on the sparsity of optimal plans \cite{friesecke2022genetic}). One can also mention the approach proposed in \cite{alfonsi2021approximation}, namely the \e{Moment Constrained Optimal Transport} (MCOT), which relies on a neat relaxation of the primal constraints (and somehow closely related to our own approach).

From a physical viewpoint, the Kantorovich dual of the MOT is a very meaningful and interesting object in its own right. In fact, as readily noticed in the literature, the so-called (not necessarily unique) Kantorovich potential can be interpreted as an external potential which forces the particles to live, at equilibrium and zero temperature, on the support of an optimal plan. Moreover, at positive temperature, the dual relates to physics still, and the (unique up to an additive constant) Kantorovich potential is interpreted, from a statistical physics viewpoint, as the external potential which forces the associated canonical ensemble to have the target density of the MOT as its one-particle density.

In this paper, we suggest to parametrize the Kantorovich potential (both at zero and positive temperature) as an electrostatic potential generated by an external charge distribution, which we call the \e{dual charge}. While such an approach had already been proposed at zero temperature in \cite{mendl2013kantorovich}, here we mainly focus on its the theoretical aspects. Moreover, from a numerical viewpoint, this dual charge seems a rather amenable candidate for discretization, as it allows for a natural ``smoothing'' from the discretization space to the space of potentials. 

\vspace{1em}

\e{\b{Acknowledgments.}}  I would like to warmly thank Mathieu Lewin (\textsc{CNRS \& Ceremade}, Université Paris Dauphine -- PSL) for having advised me during this work. This project has received funding from the European Research Council (ERC) under the European Union's Horizon 2020 research and innovation program (grant agreement MDFT No 725528).

\section{Theoretical properties on the dual charge}\label{sec:th_res}
\subsection{Transport at zero temperature}\label{sec:th_res:OT_0}
In what follows, we fix an absolutely continuous density $\rho \in L^1(\R^d, \R_+)$ with $\intRd \rho = N \in \N$ where $N \geq 2$ and $d \geq 3$, representing the expected charge density of a system of $N$ indistinguishable electrons interacting through the Coulomb potential. Among all symmetric $N$-particle probability measures $\Pro$ such that $\rho_\Pro = \rho$, where
\begin{align}\label{1_particle_density}
\rho_\Pro(\r) = N \int_{\R^{d(N-1)}} \Pro(\r, \d\r_2, \dots, \d\r_N),
\end{align}
we want to determine the one(s) yielding the lowest possible electrostatic energy, that is, we seek to solve the following minimization problem
\begin{equation}\label[problem]{prob:MOT_Coulomb}\tag{SCE}
F_{\rm SCE}(\rho) = \inf_{\rho_\Pro = \rho} \left\{ \int_{\R^{dN}} \sum_{1 \leq i < j \leq N} \frac{1}{|\r_i - \r_j|^{d-2}} \,\Pro(\d\r_1, \dots, \d\r_N) \right\}.
\end{equation}

The \pcref{prob:MOT_Coulomb} is known in the physics literature as the \e{Strictly-Correlated Electrons} (SCE) \cite{seidl1999strong, seidl1999strictly}. The functional $\rho \mapsto F_{\rm SCE}(\rho)$ arises in DFT as the semiclassical limit of the celebrated \e{Levy-Lieb functional} \cite{levy1979universal,lieb83dft, lewin2018semi}, and is used in the context of strongly-correlated systems. In particular, the SCE approach is to be thought as the exact counterpart to the \e{Kohn-Sham} (KS) approach \cite{kohn1965self}, where one seeks to map a many-body system of interacting electrons into another (fictitious) system of \e{non}-interacting electrons with the same electronic density, whereas in the SCE approach, the fictitious system is purported to have infinite electronic correlation and zero kinetic energy.  We refer the reader to the recent survey \cite{friesecke2022strong} for further details regarding \pcref{prob:MOT_Coulomb}.

As by now well-established \cite{buttazzo2012optimal, cotar2013density}, the SCE problem reformulates as a MOT problem with all marginals constrained to $\rho/N$ and cost of transportation given by the \e{Coulomb cost}
$$
c(\r_1, \dots, \r_N) =  \sum_{1 \leq i < j \leq N} \frac{1}{|\r_i - \r_j|^{d-2}}.
$$
In particular, \pcref{prob:MOT_Coulomb} is equivalent to the following maximization problem \cite{buttazzo2018continuity}, the so-called Kantorovich dual, which reads
\begin{equation}\label[problem]{prob:MOT_Coulomb_d}\tag{SCE$_D$}
F_{{\rm SCE}}(\rho) = \sup_{\substack{v  \text{ s.t.} \\\intRd |v| \rho \, < \, + \infty }} \left\{ E_N(v) + \intRd v(\r) \rho(\r)\d \r \right\},
\end{equation}
where the infimum runs over all continuous functions $v$ and where $E_N(v)$ is defined as 
$$
E_N(v) = \inf_{\r_1, \dots, \r_N} \left\{ c(\r_1, \dots, \r_N) - \sum_{i = 1}^N v(\r_i)\right\}.
$$
\begin{remark}
From a numerical viewpoint $E_N(v)$ is intractable. Instead, one can consider the following problem, which is equivalent to \pcref{prob:MOT_Coulomb_d}:
\begin{equation}\label[problem]{prob:MOT_Coulomb_d_o}\tag{SCE$_{D, \Omega}$}
F_{{\rm SCE}}(\rho) = \sup_{\substack{v \text{ s.t.} \\\int_{\Omega} |v| \rho \, < \, + \infty }} \left\{ E_{N, \Omega}(v) + \intRd v(\r) \rho(\r)\d \r \right\},
\end{equation}
where $E_{N, \Omega}(v)$ is defined for any continuous $v$ similarly to $E_N(v)$ only with the particles constrained to the support $\Omega$ of $\rho$, that is
$$
E_{N, \Omega}(v) = \inf_{\r_1, \dots, \r_N \in \Omega} \left\{ c(\r_1, \dots, \r_N) - \sum_{i = 1}^N v(\r_i)\right\}.
$$
\end{remark}
A maximizer $v$ of \pcref{prob:MOT_Coulomb_d} is called a \e{Kantorovich potential}. Notice that (minus) $v$ can be regarded as a physical potential, and that $E_N(v)$ is to be understood as the classical counterpart to the ground-state energy in quantum physics. In fact, in the physics literature, $v$ is coined as the \e{SCE potential}, which is to be thought as the effective one-body potential which emulates the SCE system, and which captures the effects of the strongly-correlated regime.

Since \pcref{prob:MOT_Coulomb} entirely pertains to electrostatics, one might conjecture that $v$ shall also relate to electrostatics. More precisely, we ask ourselves whether or not there exists a Kantorovich potential $v$ which is a \e{Coulomb potential}, that is, such that there exists a measure $\rho_{\rm ext} \in \cM(\R^d)$, which we coined as the \e{dual charge}, so that $v(\r) = \rhoext\ast|\r|^{2-d}$, where $\ast$ denotes the usual convolution operator. We would expect that $\rhoext$ verifies $\rhoext \geq 0$ and $\supp(\rhoext) \subset \supp(\rho)$, in order to attract the electrons into an optimal transport plan, which, in the absence of an external potential landscape, would escape to infinity, and that $\intRd \rhoext = N-1$, since each fixed electron ``sees'' the $N-1$ other electrons, so as to counterbalance the repulsive force created by those electrons.

\begin{figure}
\centering
  \includegraphics[width=0.7\linewidth]{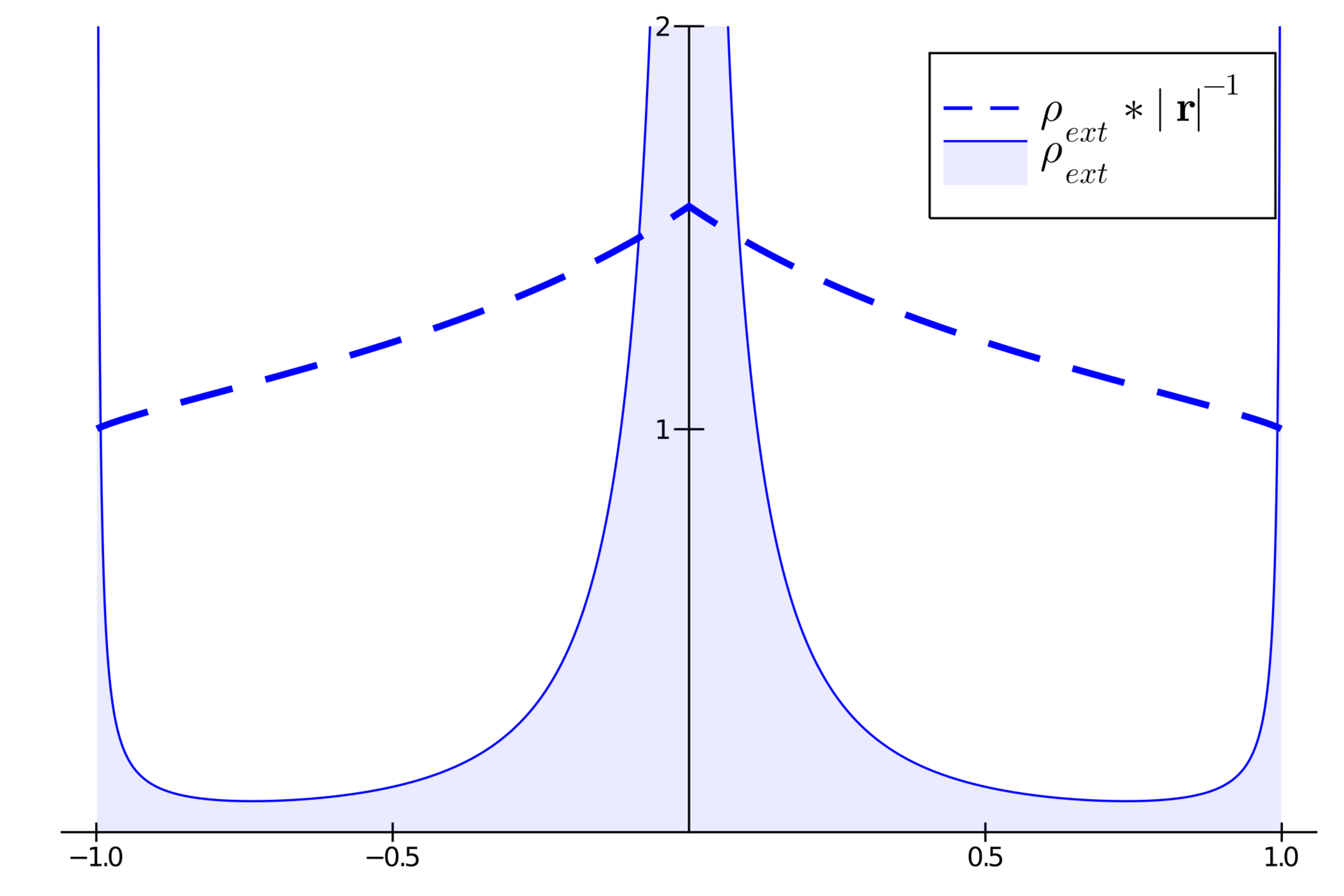}
  \label{fig:V}
\caption{Example for \Cref{thm:charge}: the (radial components of a) Kantorovich potential $\rhoext \ast |\r|^{-1}$ and its associated dual charge $\rhoext$ for a two-electron system ($N=2$) with density $\rho = 2 |B_1|^{-1} \1_{B_1}$ where $B_1$ is the unit ball of $\R^3$.}\label{fig:Vrho}
\end{figure}
The following theorem confirms our intuition:
\begin{theorem}[Existence of a dual charge at zero temperature]\label{thm:charge}
Given $\rho \in L^1(\R^d, \R_+)$ with $\intRd \rho = N \geq 2$, there exists a maximizer $v$ of \pcref{prob:MOT_Coulomb_d} which is an \e{attractive} Coulomb potential, that is of the form $v(\r) = \rhoext \ast |\r|^{2-d}$ for some positive measure\footnote{We will aways suppose that measures are \e{locally finite}.} $\rhoext$ of mass $\intRd \rhoext \leq N - 1$. Moreover, if the support $\Omega$ of $\rho$ is bounded, one can suppose that $\supp(\rhoext) \subset \Omega$ and that $\int_{\Omega} \rhoext = N-1$.\end{theorem}
The proof of \Cref{thm:charge} is consigned in \Cref{sec:proof:thm:charge}.
\begin{remark}[Non-uniqueness of dual charges]\label{rem:rhoext}
Let $\rhoext$ be a dual charge as in \Cref{thm:charge}, \e{i.e.} $v(\r) = \rhoext \ast |\r|^{2-d}$ is a Kantorovich potential for \pcref{prob:MOT_Coulomb_d}. If the support $\Omega$ of $\rho$ is bounded, \e{i.e.} $\Omega \subset B_R$ for some $R > 0$, then $\muext = \rhoext + \sigma_{R}$ is still an admissible dual charge, where $\sigma_{R}$ is the surface measure on $\partial B_R$. Indeed, $\sigma_{R}$ generates a constant potential $c_R = R^{2-d}$ inside of $B_R$ and $\sigma_R \ast |\r|^{2-d} \leq c_R$ outside of $B_R$, so that $w(\r) = \muext\ast|\r|^{2-d}$ is still a Kantorovich potential, since
\begin{equation}
\intRd v \rho + E_N(v) = \intRd w \rho + \underbrace{E_N(v + \tfrac{c_R}{N})}_{\leq E_{N}(w)}.
\end{equation}
This shows that dual charges need not be unique. Moreover, by considering $\rhoext + c\sigma_{R}$ with $c > 0$, the mass of $\rhoext$ can be made arbitrarily large. In fact, even restricting our attention to the dual charges supported on $\Omega$, the mass $\int_{\Omega} \rhoext$ can also be made arbitrarily large, by considering $\rhoext + c\chi_{\rm eq}$ where $\chi_{\rm eq} \in \cM_+(\Omega)$ is the \e{equilibrium measure} of $\Omega$ (see \cite[Chap II.]{land}). We note that generically (\e{i.e.} for << nice >> $\Omega$), $\chi_{\rm eq}$ is supported on $\partial \Omega$.
\end{remark}
\begin{remark}
In the recent survey \cite{friesecke2022strong}, it is mentioned as a conjecture \cite[Eq. (66)]{friesecke2022strong} that, when $\rho$ is supported over the entire space $\R^d$, the SCE potential $v$ shall verify the asymptotic
\begin{align}\label{asym}
v(\r) \sim \frac{N-1}{|\r|^{d-2}} \quad \text{as }\r \to \infty.
\end{align}
We remark that \Cref{thm:charge} implies that $v$ is lower than the asymptotic at \eqref{asym} — with equality in the case of a compactly supported density $\rho$, even though in this case the problem is not well-posed since $v$ can be freely modified outside of the support of $\rho$.
\end{remark}
Though the dual charge need not be unique over the entire space, we obtain uniqueness (among the class of measures which generate Lipschitz Coulomb potentials) over the support of $\rho$.
\begin{theorem}[Uniqueness of the dual charge]\label{thm:charge_unique}
Let $\rho \in L^1(\R^d, \R_+)$ with $\intRd \rho = N \geq 2$. If $\rhoext, \muext \in \cM(\R^d)$ are two dual charges for \pcref{prob:MOT_Coulomb_d} such that $\rhoext \ast |\r|^{2-d}$ and $\muext \ast |\r|^{2-d}$ are Lipschitz, then $\rhoext = \muext$ over the (topological) interior of $\Omega$, where $\Omega$ is the support of $\rho$, provided that $\Omega$ is connected.
\end{theorem}
The proof of \Cref{thm:charge_unique} is consigned in \Cref{sec:thm:charge_unique}.
\subsection{Transport at positive temperature} At positive temperature $\beta^{-1} > 0$, \e{i.e.} adding an entropic term to \pcref{prob:MOT_Coulomb}, we are led to the following variational problem
\begin{equation}\label[problem]{prob:MOT_Coulomb_temp}\tag{SCE$_\beta$}
F_{{\rm SCE}, \beta}(\rho) = \inf_{\rho_\Pro = \rho} \left\{ \int_{\R^{dN}} c(\r_1, \dots, \r_N)  \d \Pro + \beta^{-1}{\sf Ent}(\Pro| \mu^{\otimes N})\right\},
\end{equation}
where $\mu = \rho/N$ and where ${\sf Ent}(\Pro | \mu^{\otimes N}) \geq 0$ denotes the relative entropy of $\Pro$ with respect to the product measure $\mu^{\otimes N}$, defined for all probability measures $\Pro$ which are absolutely continuous with respect to $\mu^{\otimes N}$ with Radon-Nykodim density $ \frac{\d \Pro}{\d \mu^{\otimes N}}$ as
$$
{\sf Ent}(\Pro | \mu^{\otimes N}) = \int_{\R^{dN}}  \frac{\d \Pro}{\d \mu^{\otimes N}} (\r_1, \dots, \r_N) \log \left(\frac{\d \Pro}{\d \mu^{\otimes N}}(\r_1, \dots, \r_N)\right) \d \mu^{\otimes N}
$$

While \pcref{prob:MOT_Coulomb_temp} is an instance of entropy-regularized OT \cite{leonard2012schrodinger, gerolin2020multi}, which has received growing interest in the past few years, it originally pertains to the statistical mechanics of non-uniform liquids, which shares substantial connections with the DFT of classical systems \cite{wu2008density,singh1991density,evans1992density}. Indeed, the functional $\rho\mapsto F_{{\rm SCE}, \beta}(\rho)$ is the Legendre transform of the (canonical) Helmholtz free energy $\cF_\beta(v)$ with external (minus) potential $v$, 
\begin{equation}\label[problem]{prob:MOT_Coulomb_temp_d}\tag{SCE$_{D,\beta}$}
F_{{\rm SCE}, \beta}(\rho) = \sup_{v} \left\{\cF_\beta(v) + \intRd v(\r) \rho(\r) \d \r \right\}.
\end{equation}

The existence of a maximizer $v_\beta$ for \pcref{prob:MOT_Coulomb_temp_d}, which is to be thought as the effective one-body potential which forces the system into the constraint density $\rho$, was proved under relatively weak hypotheses through variational techniques in \cite[Thm 2.2.]{chayes1984inverse}. Note that $v_\beta$ is unique up to an additive constant by strict concavity, and that to respect the density constraint it must be that $v_\beta$ is infinite outside of the support of $\rho$. 

In what follows, we will make use of the two following assumptions, namely 
\begin{align}\tag{A1}\label{A1}
D(\rho) := \iint_{\R^d \times \R^d} \frac{\rho(\d \r) \rho(\d \r')}{|\r - \r'|^{d-2}} < \infty,
\end{align}
and
\begin{align}\tag{A2}\label{A2}
\Vert \rho\ast|\r|^{2-d} \Vert_{L^\infty} < \infty.
\end{align}
\begin{remark}
Appealing to \e{Hardy-Littlewood Sobolev inequality} \cite[Thm. 4.3]{lieb1996analysis}, the assumption \eqref{A1} is verified as soon as there exists $\eps > 0$ such that $\rho \in L^{\frac{2d}{d+2} - \eps}(\R^d) \cap L^{\frac{2d}{d+2} + \eps}(\R^d)$. Furthermore, it follows from \e{Hölder's inequality} that if there exists $\eps > 0$ such that $\rho \in L^{\frac{d}{2} - \eps}(\R^d) \cap L^{\frac{d}{2} + \eps}(\R^d)$, then the assumption \eqref{A2} is verified.
\end{remark}
The counterpart of \Cref{thm:charge} at positive temperature reads:
\begin{theorem}[Existence of a dual charge at positive temperature]\label{thm:charge_temp}
Let $\rho \in L^1(\R^d, \R_+)$ with $\intRd \rho = N \geq 2$ which verifies the assumptions  \labelcref{A1,A2}. Then, the unique (up to an additive constant) maximizer $v_\beta$ of \pcref{prob:MOT_Coulomb_temp_d} is an attractive Coulomb potential, \e{i.e.} there exists a positive measure $\rhoextbeta$ such that $v_\beta(\r) = \rhoextbeta \ast |\r|^{2-d}$ (up to an additive constant) on the support $\Omega$ of $\rho$. Moreover, if $\Omega$ is bounded, one can suppose that $\supp(\rhoextbeta) \subset \Omega$.
\end{theorem}
The proof of \Cref{thm:charge_temp} is consigned in \Cref{sec:thm:charge_temp}.

As one lowers the temperature, it is known \cite{carlier2017convergence} that one recovers the zero-temperature problem, \e{i.e.} \pcref{prob:MOT_Coulomb}. In the same spirit, we prove that the dual charge at positive temperature converges to a dual charge at zero temperature in the small temperature limit.
\begin{theorem}[Zero-temperature limit of the dual charge]\label{thm:charge_conv}
Let $\rho \in L^1(\R^d, \R_+)$ with $\intRd \rho = N \geq 2$ which verifies the assumptions  \labelcref{A1,A2}, and suppose that the support $\Omega$ of $\rho$ is bounded. Let $\rhoextbeta$ with $\supp(\rhoextbeta) \subset \Omega$ be as in \Cref{thm:charge_temp}, \e{i.e.} $v_\beta(\r) = \rhoextbeta \ast |\r|^{2-d}$ is the unique (up to an additive constant) maximizer of \pcref{prob:MOT_Coulomb_temp_d}. Then, the sequence $(\rhoextbeta)_{\beta > 0}$ has at least one accumulation point for the vague topology on $\cM(\Omega)$, and any of its accumulation point is an external dual charge for \pcref{prob:MOT_Coulomb_d_o}.
\end{theorem}
The proof of \Cref{thm:charge_conv} is consigned in \Cref{sec:2thm}.

\Cref{thm:charge_conv} follows from a more general result, namely the convergence at the level of the Kantorovich potentials, as consigned in \Cref{thm:vconv} hereafter. While it came to our attention that in \cite{nutz2021entropic} it is proved a similar (and decidedly more general) result, our convergence is stronger, though \e{a priori} specific to the Coulomb cost.

\begin{theorem}[Zero-temperature limit of the Kantorovich potential]\label{thm:vconv}
Let $\rho \in L^1(\R^d, \R_+)$ with $\intRd \rho = N \geq 2$ which verifies the assumptions  \labelcref{A1,A2}, and suppose that the support $\Omega$ of $\rho$ is bounded. Let $v_\beta$ be the unique maximizer of \pcref{prob:MOT_Coulomb_temp_d} which, up to an additive constant, we can suppose to verify $\cF_\beta(v_\beta) = 0$. Then, the sequence $(v_\beta)_{\beta > 0}$ is bounded in $W^{1, \infty}(\R^d)$ uniformly in $\beta$ in the limit $\beta \to \infty$, and therefore admits at least one accumulation point $v$, which is necessarily a Kantorovich potential for \pcref{prob:MOT_Coulomb_d_o}. In particular 
\begin{equation}\label{eq:Vuniform}
v_\beta \xrightarrow[\beta \to \infty]{} v \quad \text{uniformly on every compact set}.
\end{equation}
\end{theorem}
The proof of \Cref{thm:vconv} is consigned in \Cref{sec:2thm}.

\section{Numerical investigations}

In \cite{mendl2013kantorovich}, where it is introduced an analogous of our dual charge, the authors solve \pcref{prob:MOT_Coulomb_d} by discretizing the dual charge as a combination of few Gaussian functions and performing a nested unconstrained optimization. The functional derivative of $v \mapsto E_N(v)$ being intractable, they appealed to derivative-free methods for the outer optimization, while computing $E_N(v)$ through standard quasi-Newton methods. Nevertheless, as noticed by the authors, << \e{the derivative-free methods are not suitable for optimizing with respect to a large number of degrees of freedom. More efficient numerical methods need to be developed in order to obtain the Kantorovich dual solution for more general systems} >>.

We propose to approximate the solution of \pcref{prob:MOT_Coulomb_d} by considering the problem at positive temperature, which is more amenable for a computational viewpoint. Similarly to \cite{mendl2013kantorovich}, we discretize the dual charge as a combination of basis functions. As the temperature is lowered, we shall recover the zero-temperature dual charge. Note that the idea of approaching numerically the optimal transport by its entropy-regularized version has gained a lot of popularity in the recent years \cite{peyre2019computational, cuturi2013sinkhorn}. In the context of \pcref{prob:MOT_Coulomb}, this approach has already been used in \cite{benamou2016numerical, nenna2016numerical}. 
\subsection{Discretization of the dual charge}
We give ourselves a finite set of $M$ basis functions $B = \left\{\rho_i \right\}_{i = 1, \dots, M}$ where $\rho_i \in L^1(\R^d, \R_+)$ for all $i = 1, \dots, M$, and we seek to approximate the external dual charge of \pcref{prob:MOT_Coulomb_temp_d} as a linear combination of the $\rho_i$'s, that is by $\rho_\ext[\nu]$ for some set of weights $\nu = (\nu_i)_{i= 1, \dots, M}$, where
\begin{equation}\label{eq:discr}
\rho_\ext[\nu] := \sum_{i = 1}^{M} \nu_{i} \rho_i, \quad \nu \in \R^{M}.
\end{equation}
Otherwise stated, we approximate the Kantorovich potential $v_\beta$ of \pcref{prob:MOT_Coulomb_temp_d} by $v[\nu]$ where $v[\nu](\r) = \rho_\ext[\nu]\ast|\r|^{2-d}$. We are then lead to the following optimization problem
\begin{equation}\label[problem]{prob:MOT_d}\tag{D$_{B}$}
G_{\SCE, \beta}(\rho; B) = \sup_{\nu \in \R^M}  \left\{ \cF_\beta(v[\nu]) + \intRd v[\nu](\r)\rho(\r) \d \r \right\}.
\end{equation}
\begin{remark}
If we define the set of functions $B' = \{\rho_i \ast |\r|^{2-d}\}_{i = 1, \dots, M}$, one can show that the following duality for \pcref{prob:MOT_d} holds
$$
G_{\SCE, \beta}(\rho; B) = \inf_{\Pro \in \Pi(\rho; B') } \left\{ \int_{\R^{dN}} c(\r_1, \dots, \r_N) \d \Pro + \beta^{-1} {\sf Ent}(\Pro | \mu^{\otimes N}) \right\},
$$
where as previously $\mu = \rho/N$ and where $\Pi(\rho; B')$ is defined as the set of $N$-particle probability measures $\Pro$ verifying the following moment constraints for all $\phi \in B'$ 
\begin{align}\label{relax_constraint}
\int_{\R^{dN}} \sum_{i = 1}^N \phi(\r_i) \d \Pro(\r_1, \dots, \r_N) = \intRd \phi(\r) \rho(\r) \d \r.
\end{align}
Therefore \pcref{prob:MOT_d} can be regarded as the dual approach to the \e{Moment Constrained Optimal Transport} (MCOT) as introduced in \cite{alfonsi2021approximation}.
\end{remark}
\begin{remark}
According to \Cref{thm:charge_temp}, the dual charge is a positive measure. Moreover, in the small temperature limit the dual charge can be assumed to have total mass no greater than $N-1$ according to \Cref{thm:charge}. Therefore, it seems natural to consider the following constrained version of \pcref{prob:MOT_d}
$$
\sup_{\nu \in \Delta_{B}}  \left\{ \cF_\beta(v[\nu]) + \intRd v[\nu](\r)\rho(\r) \d \r \right\}
$$
where $\Delta_B$ is the convex set defined as
$$
\Delta_{B} = \left\{ \nu \in \R_+^{M} : \sum_{i = 1}^{M} \nu_i  \int_{\R^d} \rho_i  \leq N-1 \right\}.
$$
Nevertheless, the optimum $\nu^*$ of \pcref{prob:MOT_d} need not be in $\Delta_{B}$. In fact, one should dismiss this additional constraint when dealing with rather crude discretizations, as it might yield far-from-optimal results. 
\end{remark}

Denoting by $\G[\nu]$ the objective of \pcref{prob:MOT_d}, we note that $\nu \mapsto \mathbb{G}[\nu]$ is strictly concave. We can solve \pcref{prob:MOT_d} using a classical steepest ascent maximization algorithms. A straightforward computation shows that, for all $i = 1, \dots, M$, we have
\begin{equation}\label{eq:gradient}
\partial_{\nu_i} \mathbb{G}[\nu] = D(\rho_i, \rho - \rho[\nu]) 
\end{equation}
where $\rho[\nu]$ denotes the one-particle density of the configurational canonical ensemble with external potential (minus) $v[\nu]$, and where for any two measures $\mu, \nu \in \cM(\R^d)$ we define
$$
D(\mu, \nu) = \intRd \frac{\mu(\d \r) \nu(\d \r')}{|\r - \r'|^{d-2}}.
$$ 

From a numerical perspective, the quantity $D(\rho_i, \rho)$ can easily be computed for all $i = 1, \dots, M$. In fact, all the computational burden is hidden in the computation of $\rho[\nu]$. Nevertheless, we notice that 
\begin{align}
D(\rho_i, \rho[\nu]) = \left\langle\sum_{j = 1}^N \rho_i \ast |\r_j|^{2-d} \right\rangle_{\beta, \nu}
\end{align}
where $\langle \cdot \rangle_{\beta, \nu}$ denotes the expectation value with respect to the canonical ensemble. Therefore, we can appeal to Monte-Carlo (MC) methods. At high temperature (and in the compactly-supported case), one can resort to vanilla MC with uniform proposals over the support of $\rho$. At low temperature, as the ensemble crystallizes onto a minimizer of \pcref{prob:MOT_Coulomb}, one needs to resort to fancier sampling methods since minimizers are believed to be generically singular. For the canonical ensemble is the unique invariant ergodic measure of the (\e{overdamped}) \e{Langevin diffusion}, we can for example produce approximate samples of the ensemble through a myriad of algorithms \cite{rousset2010free, cances2007theoretical}.

\subsection{Numerical investigations}
\subsubsection{One-dimensional system} We start, as a sanity check, with a one-dimensional system, \e{i.e.} $d=1$. This problem is amenable for numerical investigations, as it is entirely solvable for all $N \geq 2$, as stated in \Cref{thm:1d_fwc} below. Recall that, in one-dimension, the Coulomb potential is given by $-|\r - \r'|$. We have not considered this specific case (nor the two-dimensional case where the Coulomb potential is given by  $-\ln|\r - \r'|$) in our main theorems to avoid potential problems due to the divergence of both potentials at infinity. Nevertheless, we expect similar results to hold (see \Cref{rmk:extension1d} below).

\begin{theorem}[{\cite[Thm. 1.1]{colombo2015multimarginal}}]\label{thm:1d_fwc}
Let $\rho \in L^1(\R, \R_+)$ with $\int_{\R} \rho = N \geq 2$, and suppose that $\int_{\R} |r| \rho < \infty$. Then, there exists an explicit measurable map $t : \R \to \R$ such that, denoting $t^{(j)}$ the $j$-th fold composition of $t$ with itself, the probability measure 
\begin{align}\label{monge_type}
\Pro_t(r_1, \dots, r_N) = \frac{ \rho(r_1) }{N}\otimes \delta(r_2 - t(r_1)) \otimes \cdots \otimes \delta(r_N - t^{(N-1)}(r_1))
\end{align}
is a minimizer for the \pcref{prob:MOT_Coulomb} in dimension $d = 1$, that is 
$$
\inf_{\rho_\Pro = \rho} \left\{ - \int_{\R^N} \sum_{1 \leq i < j \leq N} |r_i - r_j| \d \Pro(r_1, \dots, r_N)\right\}.
$$
Moreover, given $\ell_1 < \dots < \ell_{N-1}$ such that $\int_{\ell_i}^{\ell_{i+1}} \rho = 1$ for all $i = 0, \dots, N$, with $\ell_0 = -\infty$ and $\ell_N = +\infty$, the function $v$ defined as
$$
v(r) = - \left( \sum_{i = 1}^{N-1} \delta_{\ell_i} \right) \ast |r|
$$
is a Kantorovich potential for the  \pcref{prob:MOT_Coulomb_d} in $d=1$. In particular, the comb $\sum_{i = 1}^{N-1} \delta_{\ell_i} $ is, according to our terminology, a dual charge, which only depends on $\rho$ through the $\ell_i$'s.
\end{theorem}
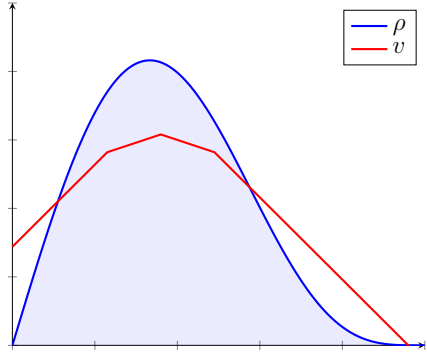
\begin{figure}
\begin{tikzpicture}[scale = 0.8]
\begin{axis}[
    axis lines = left,
    ymin=0,
    ymax=2.5,
    yticklabels={,,},
    xticklabels={,,}
]
\addplot [
    domain=0:1, 
    samples=1000, 
    color=blue,
    line width = 1,
    fill = blue,
    fill opacity = 0.08
]
{10*x*(1-x^2)^4};
\addlegendentry{$\rho$}

\addplot [
    domain=0:1, 
    samples=1000, 
    color=red,
    line width = 1,
]
{-abs(x - 0.23) - abs(x-0.36) - abs(x-0.49) +1.8};

\addlegendentry{$v$}
\end{axis}
\end{tikzpicture}

\caption{Example for \Cref{thm:1d_fwc} : $\rho(r) \propto r(1-r^2)^4$ on $[0,1]$ and the corresponding (up to an additive constant) optimal Kantorovich potential $v$ for $N = 4$.}
\end{figure}

\begin{remark}
The result in \cite[Thm. 1.1]{colombo2015multimarginal} does not \e{per se} apply to the one-dimensional Coulomb potential as claimed above, for it is not strictly convex. Nevertheless, it is straightforward to prove that $t$ remains an optimal transport plan which in this case need not be unique. In fact, there might exist optimal plans which are not of the form \eqref{monge_type} and which carry entropy. For instance, for any $N \geq 2$ and $\rho = \1_{[-N/2,N/2]}$, the following probability measure 
\begin{align}
\Pro_N = \frac{1}{N!}\sum_{\substack{i_1, \dots, i_N \,=\, 0, \dots, N-1\\ i_k \neq i_l \text{ for } k \neq l}} \1_{[\ell_{i_1}, \ell_{i_1+1}]}(r_1) \otimes \cdots \otimes \1_{[\ell_{i_N}, \ell_{i_N +1}]}(r_N) 
\end{align}
is optimal. We conjecture that among all minimizers of \pcref{prob:MOT_Coulomb} for this given density $\rho$, $\Pro_N$ is the one which maximizes the entropy. Therefore, if one considers the one-dimensional problem at positive temperature $\beta^{-1}$, then in the small temperature limit $\beta \to \infty$, one shall recover $\Pro_N$. \end{remark}
\begin{remark}
We see from \Cref{thm:1d_fwc} that in one-dimension \pcref{prob:MOT_Coulomb_d} is degenerate with respect to $\rho$. Indeed, the maximizer $v$ only depends on the positions $\ell_i$'s, and not on the ``shape'' of the density $\rho$ inside each segment $[\ell_i, \ell_{i +1}]$. This pertains to the special case of the one-dimensional Coulomb force, which does not depend on the distance between the particles — \e{i.e.} when the electrons are consigned to their respective segment, the potential felt by them is effectively constant, so that their exact positions inside the segments are irrelevant.
\end{remark}
\begin{figure}
\centering
  \includegraphics[width=.9\linewidth]{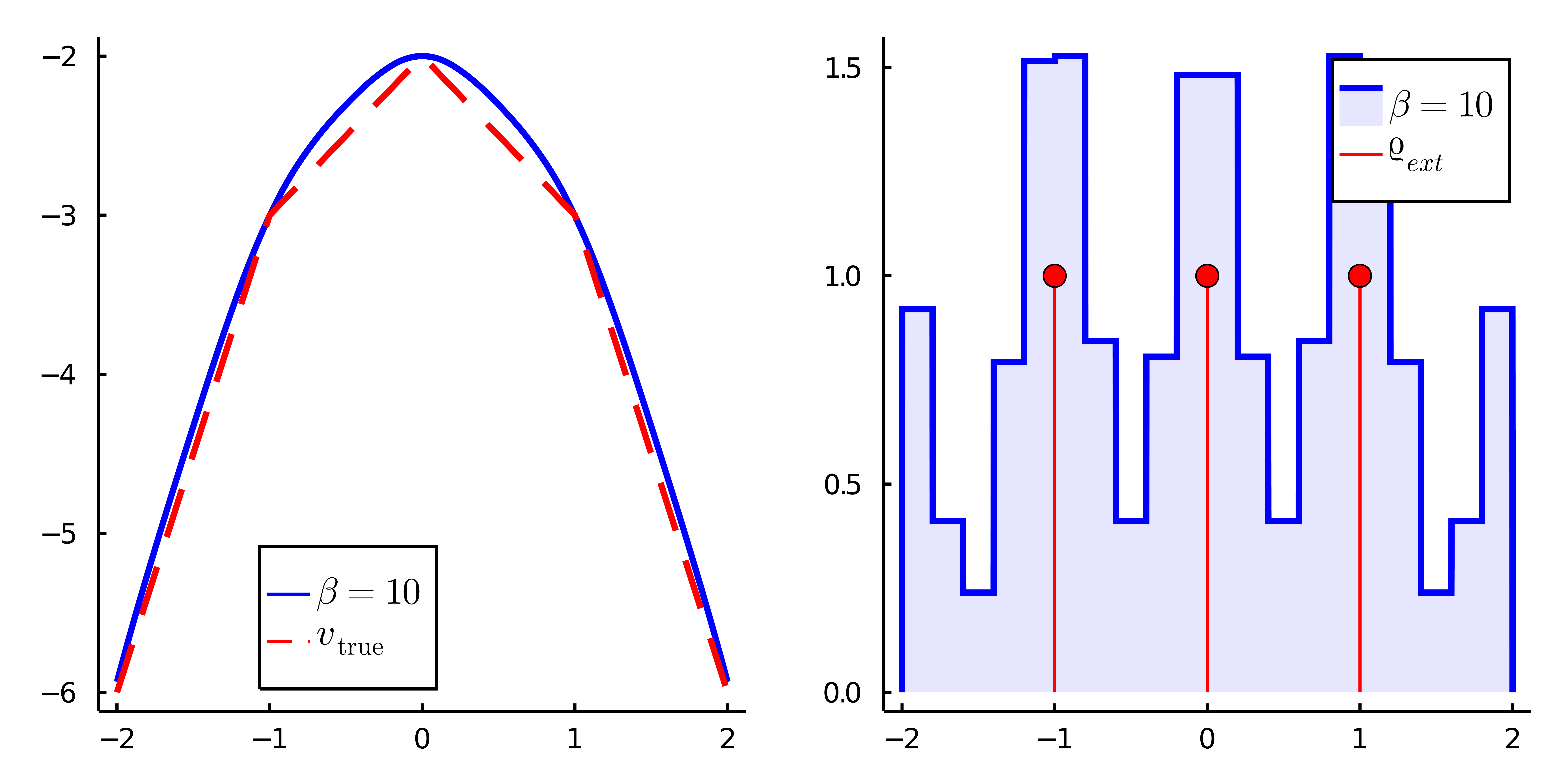}
\caption{One-dimensional four-electron droplet, \emph{i.e.} $\rho = \1_{[-2,2]}$, and $\beta = 10$. \e{Left}: the approximate Kantorovich potential $v[\nu]$ for $\nu$ obtained with our algorithm compared to the exact Kantorovich potential given by \Cref{thm:1d_fwc}. \e{Right}: the associated dual charge compared to the exact dual charge.}  \label{fig:N1}
\end{figure}
We run our algorithm with $N=4$ at inverse temperature $\beta = 10$ with density $\rho = \1_{[-2,2]}$. We select a crude discretization where the elements of $B$ are indicator functions of evenly-spaced segments of width $2/M$, that is
$$
\rho_i = \1_{\left[i\frac2M, (i+1)\frac2M\right]} \quad \text{for all }i=0, \dots, M-1
$$
with $M = 10$, the segment $[-2, 0]$ being dealt with symmetrically. The results obtained, which are displayed at \Cref{fig:N1}, are consistent with \Cref{thm:1d_fwc}.

\begin{remark}[Extension of the main theorems to the one-dimensional case]\label{rmk:extension1d}
In one-dimension, we see that \Cref{thm:charge}, which is stated only for $d \geq 3$, remains valid according to \Cref{thm:1d_fwc}. Similar uniqueness results as in \Cref{thm:charge_unique} holds in one-dimension. At positive temperature $\beta^{-1}$, the existence of a dual charge (\e{i.e.} \Cref{thm:charge_temp}) also remains valid. In fact, we even have a stronger result in the one-dimensional case, namely that the dual charge $\rho_{\ext, \beta} = -\frac12 v_\beta''$ at positive temperature has mass $N-1$, as in the zero-temperature case. Indeed, suppose that the support $\Omega$ of $\rho$ is connected, that is of the form $[a, b]$. Then, we have 
\begin{align}
-\frac{1}{2}\int_{[a,b]} v_\beta''(r) \d r = \frac12(v_\beta'(a) - v_\beta'(b)).
\end{align}
Now, appealing to the equation verified by $v_\beta$ (see \Cref{eq:Vbetaeq} below in the case $d\geq3$), we have
\begin{align}
v_\beta'(r) = - \int_{\R^{N-1}} \sum_{i = 2}^N {\rm sgn}(r - r_i) \d \overline{G_\beta^r}(r_2, \dots, r_N),
\end{align}
where $\overline{G_\beta^r}$ is the (normalized) Gibbs measure defined as in \eqref{def:Gbr} below. Therefore, we have $v_\beta'(a) = N-1$ and $v_\beta'(b) = -(N-1)$, leading to the claimed fact. Finally, regarding the convergence of the Kantorovich potential (\Cref{thm:vconv}) and the dual charge (\Cref{thm:charge_conv}) at positive temperature in the small temperature limit $\beta \to \infty$, these results also remain veracious in one-dimension (using similar arguments).
\end{remark}
\subsubsection{Three-dimensional droplets} We now consider uniform droplets in three-dimension, that is, $\rho = N |B_1|^{-1} \1_{B_1}$ where $B_1$ is the unit ball of $\R^3$. Uniform droplets were numerically investigated in \cite{rasanen2011strictly} up to $N = 30$ using the now called \e{Seidl's maps}, which are believed to furnish a minimizer of \pcref{prob:MOT_Coulomb} and which, irrespective of its conjectured optimality, is known to be near-optimal. In the special two-electron case, the problem is analytically solvable. The unique minimizer of \pcref{prob:MOT_Coulomb} is known to be
\begin{align}
\frac{\rho(\r_1)}{2} \otimes \delta(\r_2 - t(\r_1)) \quad \text{ where }\,\, t(\r) = - \frac{\r}{|\r|}\left(1 - |\r|^3\right)^{1/3}.
\end{align}
Given a Lipschitz Kantorovich potential $v$ (which always exists; see \Cref{duality_zero}), we have for almost-every $\r \in B_1$
\begin{equation}\label{ex:eqV}
\nabla v(\r) = - \frac{\r - t(\r)}{|\r - t(\r)|^3}.
\end{equation}
Then $v$ can be retrieved by integrating \eqref{ex:eqV}. Conversely, the dual charge associated to $v$ is found to be $-\Delta v(\r)/4\pi$ where $\Delta$ denotes the usual Laplacian operator, that is, for almost-every $\r\in B_1$
\begin{align}\label{ex:eqlap}
\rhoext(\r) = \frac{(4\pi)^{-1}}{|\r|(1 -|\r|^3)^{\frac23} (|\r| + (1 - |\r|^3)^{\frac23})^3}.
\end{align}

We run our algorithm with $N = 2$ and $\rho = 2|B_1|^{-1}\1_{B_1}$ at decreasing temperatures. The dual charge is discretized into $M = 15$ evenly-spaced concentric shells in which the charge is constant. The results obtained at \Cref{fig:N2} are consistent. We then select a fixed temperature $\beta = 50$ and we run our algorithm for an increasing number of particles up to $N = 30$. We use the same discretization as for the two-particle case, except when $N=30$, in which case $M$ is raised up to $25$. Given the discretized version of \pcref{prob:MOT_d} at zero-temperature, that is
$$
F_{\SCE} [\nu] = E_N(v[\nu]) + \intRd v[\nu](\r) \rho(\r) \d \r,
$$
we can assess the performance of our procedure by comparing $F_{\SCE}[\nu]$ with the values obtained in \cite{rasanen2011strictly}, which are known to be (near-)optimal. The results obtained are displayed at \Cref{fig:droplet}. We notice that, even appealing to rather crude discretization, our algorithm yields near-optimal values. The case $N=30$ was obtained in $\sim 6$ min on a personal computer using the programming language \texttt{Julia}. 
\begin{figure}
\centering
  \includegraphics[width=.9\linewidth]{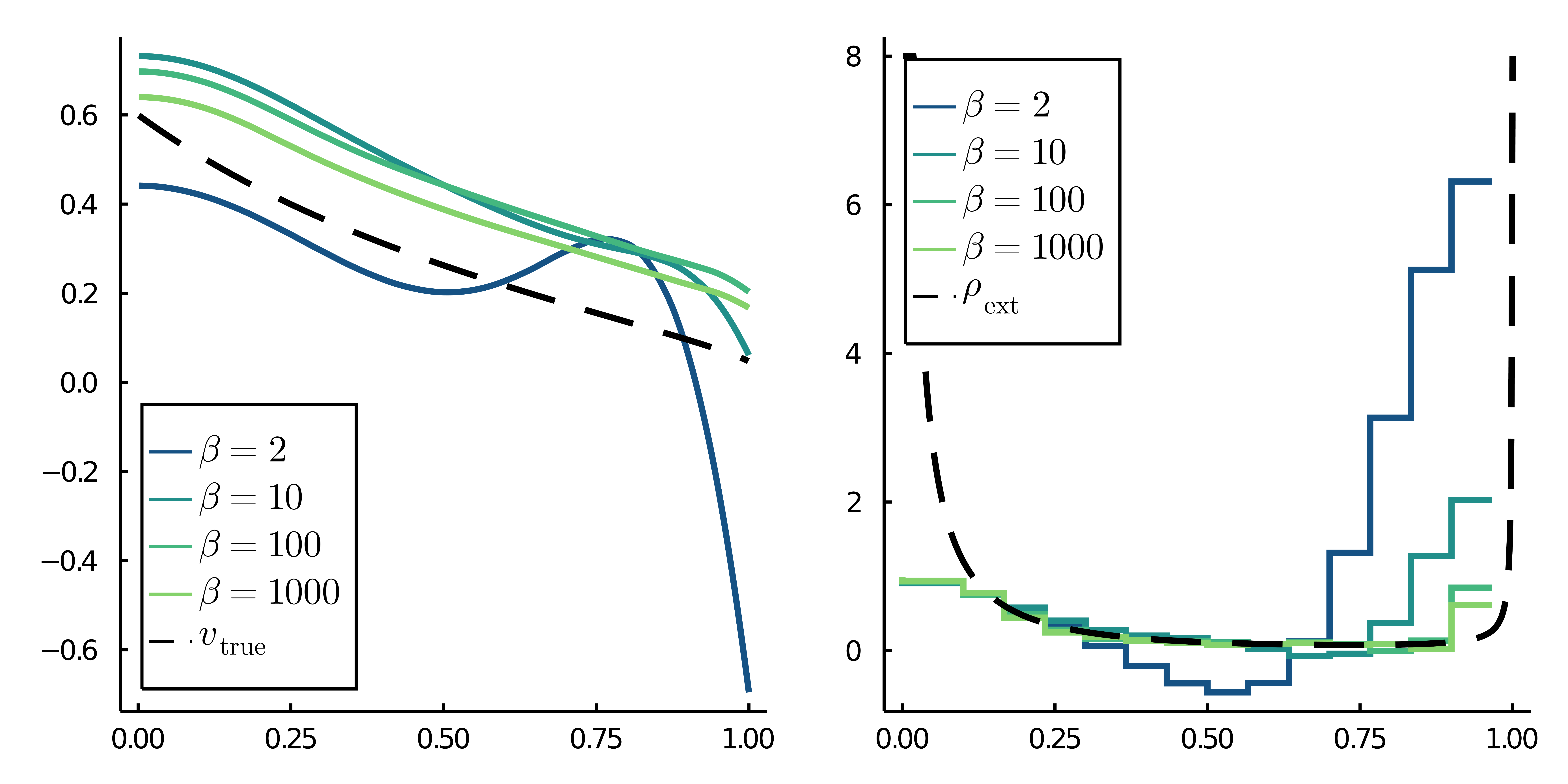}
\caption{Two-electron droplet, \emph{i.e.} $\rho = 2|B_1|^{-1}\1_{B_1}$ where $B_1$ is the unit ball of $\R^3$. \e{Left}: the Kantorovich potentials $v[\nu]$ for $\nu$ obtained with our algorithm at decreasing temperatures compared to the exact Kantorovich potential (defined up to an additive constant). \e{Right}: the associated dual charges for the corresponding temperatures compared to the exact dual charge. }\label{fig:N2}
\end{figure}

\begin{figure}[t]
  \begin{minipage}{0.49\textwidth}
    \includegraphics[trim=0 250 0 0, width=\textwidth, height=6cm]{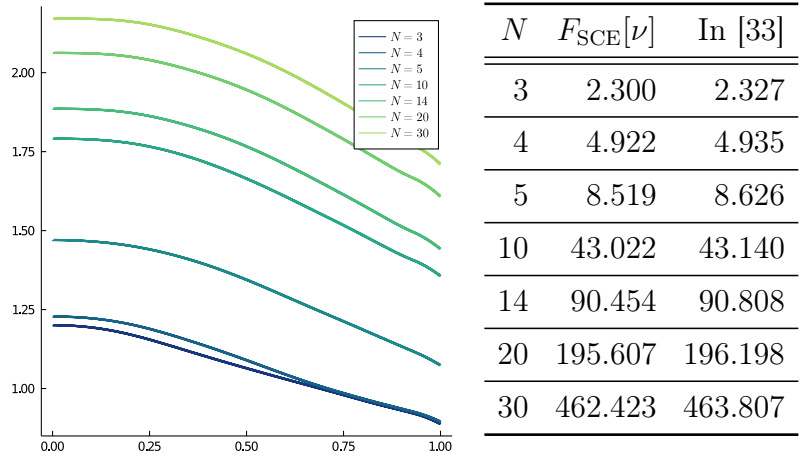}
  \end{minipage}
  \begin{minipage}{0.35\textwidth}
  \begin{center}
  \begin{tabular}{rrrr} 
	\toprule[1pt]
  $N$ & $F_{\SCE}[\nu]$ & In \cite{rasanen2011strictly}  \\ 
  \midrule\midrule
  3 & 2.300 & 2.327  \\ 
\midrule
  4 & 4.922 & 4.935 \\ 
\midrule
  5 & 8.519 & 8.626   \\ 
\midrule
  10 & 43.022 & 43.140   \\
\midrule
  14 & 90.454 & 90.808   \\
\midrule
  20 & 195.607 & 196.198   \\
\midrule
  30 & 462.423 & 463.807   \\
  \bottomrule[1pt]
\end{tabular}
\end{center}
  \end{minipage}
      \caption{Droplets, \emph{i.e.} $\rho = N|B_1|^{-1}\1_{B_1}$ where $B_1$ is the unit ball of $\R^3$. The inverse temperature is fixed to $\beta = 50$, and the number of electrons increases until $N = 30$. We compare the total energy of the system at zero temperature with the values in \cite{rasanen2011strictly}. We also display the approximate Kantorovich potentials obtained, macroscopically rescaled \e{i.e.} $v[\nu]/N$. }\label{fig:droplet}
\end{figure}

\subsection{Implementation and bottlenecks}\label{discussions}
Given the gradient $\nabla \G[\nu_t]$ of the objective at the iteration point $\nu_t$, the outer optimization in \pcref{prob:MOT_d} can be conducted using first-order or quasi-Newton methods. In our experiments, we used a \e{Gradient Ascent} algorithm with fixed step-size and added momentum, more precisely the \emph{Nesterov’s Accelerated Gradient} (NAG) procedure. The objective $\G$ being extremely expansive to compute because of the free energy term $\cF_\beta$, classical line search procedures are impracticable. Note that it would be interesting to implement a procedure which bypasses the usual line searches algorithms to only enforce approximate orthogonality of the gradient with the search direction.

The inner optimization for \pcref{prob:MOT_d} consists in determining the gradient $\nabla \G[\nu_t]$. As mentioned above, this can done using many different algorithms, the simplest of them being the \emph{Unadjusted Langevin Algorithm} (ULA; see \cite{roberts1996exponential}). Though the original ULA is not \e{a priori} tailored for compactly supported densities, we found it to perform rather well — at least in the case of the droplets. Evidently, an important bottleneck is that, as the temperature is decreased, it becomes harder to sample from the canonical ensemble.

\section{Proofs}\label{sec:dual_ext}

\subsection{On duality theory at zero temperature}\label{duality_zero}We briefly recall some important facts regarding the duality theory for \pcref{prob:MOT_Coulomb}. In \cite{buttazzo2018continuity, colombo2019continuity}, it is proved that the electrons cannot overlap at optimality, in the sense that there exists $\alpha > 0$ such that any minimizer $\Pro$ of \pcref{prob:MOT_Coulomb} is supported away from $D_\alpha$, \e{i.e.} $\Pro(D_\alpha) = 0$, where
\begin{align}\label{diag}
D_\alpha = \{ (\r_1, \dots, \r_N) : \exists i \neq j \text{ such that } |\r_i - \r_j| \leq \alpha \}.
\end{align}
In particular, one can substitute to $c$ the \e{truncated} Coulomb cost $c_\alpha$, where
\begin{align}
c_\alpha(\r_1, \dots, \r_N) = \sum_{1 \leq i < j \leq N} \min\left\{ \frac{1}{\alpha^{d-2}},  \frac{1}{|\r_i - \r_j|^{d-2}}\right\}.
\end{align}
Using this fact, one can prove \cite[Thm. 2.6]{buttazzo2018continuity} that \pcref{prob:MOT_Coulomb} admits the following dual formulation, which is equivalent to  \pcref{prob:MOT_Coulomb_d}:
\begin{align}\label[problem]{dual}\tag{D$_\alpha$}
F_{\rm SCE}(\rho) = \max_{\substack{v \text{ s.t.} \\\intRd |v| \rho \, < \, + \infty }} \left\{ E_{N, \alpha}(v) + \intRd v(\r) \rho(\r) \d \r \right\},
\end{align}
where $E_{N, \alpha}(v)$ is defined similarly to $E_N(v)$, only with the truncated cost $c_\alpha$ substituted to $c$, that is
\begin{align}\label{Eva}
E_{N, \alpha}(v) = \inf_{\r_1, \dots, \r_N} \left\{ c_\alpha(\r_1, \dots, \r_N) - \sum_{i = 1}^N v(\r_i) \right\}.
\end{align}
For brevity, we drop the subscript $\alpha$ in \eqref{Eva} from now on. Moreover, according to \cite[Lem. 3.3 and Thm. 3.6]{buttazzo2018continuity}, there exists a maximizer $v$ of \pcref{dual} which verifies
\begin{align}\label{eq_v}
v(\r) = \inf_{\r_2, \dots, \r_N} \left\{ c_\alpha(\r, \r_2, \dots, \r_N) - \sum_{i = 2}^N v(\r_i) \right\}.
\end{align}
Note that, for a Kantorovich potential $v$ which verifies \eqref{eq_v}, it holds that $E_N(v) = 0$. Moreover, we have the following lemma.
\begin{lemma}\label{limit_v}
Let $v$ be a Kantorovich potential for \pcref{prob:MOT_Coulomb_d} verifying \Cref{eq_v}. Then, $v$ is Lipschitz and uniformly bounded. Moreover, it satisfies the following limit
\begin{align}
\lim_{|\r| \to \infty} v(\r) = E_{N-1}(v) < 0
\end{align}
where, according to \eqref{Eva}, we have
$$
E_{N-1}(v) = \inf_{\r_2, \dots, \r_N} \left\{ c_\alpha(\r_2, \dots, \r_N) - \sum_{i = 2}^N v(\r_i) \right\}.
$$
\end{lemma}
\begin{proof}[Proof of \Cref{limit_v}]
From \cite[Thm. 3.4 and Thm. 3.6]{buttazzo2018continuity}, we know that $v$ is uniformly bounded and Lipschitz. Now, using \Cref{eq_v}, for all $\r$ we have
\begin{align}\label{v_low}
v(\r) \geq \inf_{\r_2, \dots, \r_N} \left\{ c_\alpha(\r_2, \dots, \r_N) - \sum_{i = 2}^N v(\r_i) \right\} = E_{N-1}(v).
\end{align}
Moreover, for all $\r$ and $(\r_2, \dots, \r_N)$, we have
\begin{align}\label{ff}
v(\r) \leq c_\alpha(\r, \r_2, \dots, \r_N) - \sum_{i = 2}^N v(\r_i),
\end{align}
which entails that
\begin{align}\label{ls}
\limsup_{|\r| \to \infty} v(\r) \leq c_\alpha(\r_2, \dots, \r_N) - \sum_{i = 2}^N v(\r_i).
\end{align}
Taking the infimum with respect to $\r_2, \dots, \r_N$ then leads to the conclusion that $\lim_{|\r| \to \infty} v(\r) = E_{N-1}(v)$. Then, letting $|\r_i| \to \infty$ for all $i = 2, \dots,N$ in \eqref{ff}, we obtain
\begin{align}
E_{N-1}(v) \leq -(N-1)E_{N-1}(v),
\end{align}
which implies that $E_{N-1}(v) \leq 0$. If $E_{N-1}(v) = 0$, then, once again appealing to \eqref{v_low} and \eqref{ff} and letting $|\r_i| \to \infty$ for $i = 2, \dots, N$, we would obtain $v\equiv 0$, which is impossible. Therefore $E_{N-1}(v) < 0$.
\end{proof}

\subsection{Proof of \Cref{thm:charge}}\label{sec:proof:thm:charge}
Let us turn to the proof of \Cref{thm:charge} regarding the existence of a dual charge at zero temperature. As mentioned above, there exists a Kantorovich potential $v$ for \pcref{prob:MOT_Coulomb_d} which verifies 
\begin{equation}\label{eq:t0}
v(\r) = \inf_{\r_2, \dots, \r_N} \left\{ c_\alpha(\r, \r_2, \dots, \r_N) - \sum_{i = 2}^N v(\r_i) \right\}.
\end{equation}
We will show that this particular potential arises from a charge. Indeed, for any $\r_2, \dots, \r_N$, the function
$$
\r \mapsto  c_\alpha(\r, \r_2, \dots, \r_N) - \sum_{i = 2}^N v(\r_i)
$$
is superharmonic (see {\cite[Chap I.2]{land}}) and uniformly Lipschitz in the $\r_i$'s. Therefore, any $v$ which verifies \eqref{eq:t0} remains superharmonic as the pointwise infimum of a set of uniformly Lipschitz superharmonic functions. Now, we recall the following theorem:
\begin{theorem}[Riesz's decomposition theorem {\cite[Thm 1.24]{land}}]\label{thm:riesz}
Let $f : \R^d \to \R$ be a superharmonic function which admits a harmonic minorant, \e{i.e.} there exists a harmonic function $m : \R^d \to \R$ with $m(\r) \leq f(\r)$ for all $\r \in \R^d$. Then, there exists a positive measure $\mu \in \cM_+(\R^d)$ and a harmonic function $h : \R^d \to \R$  such that 
\begin{equation}\label{eq:riesz}
f(\r) = \mu \ast |\r|^{2-d} + h(\r) \quad \text{for all } \r\in \R^d.
\end{equation}
\end{theorem}
We are now ready to provide the proof of \Cref{thm:charge}.
\begin{proof}[Proof of \Cref{thm:charge}]
Let us consider $v$ a Kantorovich potential for \pcref{prob:MOT_Coulomb_d} which verifies \eqref{eq:t0}. As noted above, $v$ is superharmonic. Moreover, according to \Cref{limit_v}, $v$ is uniformly bounded. Therefore, according to \Cref{thm:riesz} above, there exists a positive measure $\rho_\ext \in \cM_+(\R^d)$ and an harmonic function $h$ such that $v(\r) = \rho_\ext \ast |\r|^{2-d} + h(\r)$. But $h$ must be bounded from above, so that \e{Liouville's theorem} entails that $h$ is constant. This provides the existence of a dual charge as stated by \Cref{thm:charge}. In fact, let us prove that the constant is provided by $E_{N-1}(v)$ the limit of $v$ at infinity (see \Cref{limit_v}). This is not entirely trivial, as the limit of a Coulomb potential need not to be zero. Nevertheless, we have the following result, whose proof is given below.
\begin{lemma}\label{limit_vc}
Let $\mu \in \cM_+(\R^d)$ be any locally finite measure such that $\mu \ast |\r|^{2-d}$ is bounded at the origin, \e{i.e.} $\intRd |\r|^{2-d}\mu(\d\r) < \infty$. Then 
\begin{align}
\liminf_{|\r| \to \infty} \mu \ast |\r|^{2-d} = 0.
\end{align}
\end{lemma}
In fact, because $v$ admits a limit at infinity, \Cref{limit_vc} implies the stronger result that $\rhoext \ast |\r|^{2-d}$ converges to zero everywhere at infinity
\begin{align}\label{limit_vc0}
\lim_{|\r| \to \infty} \rhoext \ast |\r|^{2-d} = 0.
\end{align}
Let us prove that the total mass of the dual charge $\rhoext$ cannot exceed $N-1$. In what follows, we denote $v_\ext(\r) = \rhoext \ast |\r|^{2-d}$. Note that, according to \eqref{limit_vc0} and \Cref{limit_v}, we have for all $\r \in \R^d$
\begin{equation}\label{eq:vvext}
v(\r) = v_\ext(\r) + E_{N-1}(v).
\end{equation}
Moreover, for $E_N(v) = 0$, we have the following inequality for all $\r$ and $(\r_2, \dots, \r_N)$
\begin{align}\label{def:step1}
v_\ext(\r) \leq c_{\alpha}(\r, \r_2, \dots, \r_N) - \sum_{i = 2}^N v_\ext(\r_i)- E_{N}(v_\ext).
\end{align}
We first claim that $E_{N}(v_\ext) = E_{N-1}(v_\ext)$. Indeed, according to \eqref{eq:vvext} and still appealing to the fact that $E_N(v) = 0$, we have
\begin{equation}
E_N(v_\ext) = N E_{N-1}(v).
\end{equation}
Furthermore, we have
\begin{equation}
E_{N-1}(v) = E_{N-1}(v_\ext) - (N-1)E_{N-1}(v).
\end{equation}
This entails that $N E_{N-1}(v) = E_{N-1}(v_\ext)$, and ultimately implies the claim that $E_{N}(v_\ext) = E_{N-1}(v_\ext)$. Now, because $v_\ext$ vanishes at infinity, we have the following sequence of inequalities
\begin{align}
E_{N-K}(v_\ext) \leq E_{N-K+1}(v_\ext) \quad \text{for all } K = 1, \dots, N-1,
\end{align}
where the limiting case $K = N-1$ corresponds to the case where all electrons but one are sent to infinity, that is
$$
E_{1}(v_\ext) = \inf_{\r} \left\{-v_\ext(\r)\right\}.
$$
As the supremum of $v$ cannot be attained at infinity, we deduce the existence of a lowest $K \in \{1, \dots, N-1\}$ such that 
\begin{align}
E_{N-1}(v_\ext) = E_{N-2}(v_\ext) = \dots = E_{N-K}(v_\ext) < E_{N-K-1}(v_\ext).
\end{align}
The last inequality implies that the infimum $E_{N-K}(v_\ext)$ is attained inside a compact set for some $\hat{\r}_{1}, \dots, \hat{\r}_{N - K }$. Now, plugging $\hat{\r}_{1}, \dots, \hat{\r}_{N - K}$ into \eqref{def:step1}, we obtain
\begin{align}
v_\ext(\r) \leq \sum_{k = 1}^{N-K} \frac{1}{|\r - \hat{\r}_i|^{d-2}}, \quad \text{for all } \r \in \R^d,
\end{align}
which entails, as claimed in \Cref{thm:charge}, that 
\begin{align}
\intRd \rhoext \leq N-K \leq N-1.
\end{align}

Finally, for the last item of \Cref{thm:charge}, namely that we can suppose $\supp(\rhoext) \subset \Omega$ in the case where $\Omega$ is bounded, it suffices to consider the \e{balayage} of $\rhoext$ onto $\Omega$, as stated in the following technical theorem.
\begin{theorem}[Balayage {\cite[Thm 4.2, Thm 4.4]{land}}]\label{thm:bal} Let $G \subset \R^d$ be a region with compact boundary. Given any $\mu \in \cM_+(\R^d)$ such that $\supp(\mu) \subset \overline{G}$, there exists a measure $\nu \in \cM_+(\R^d)$, the so-called \e{balayage measure of $\mu$ onto $\partial G$}, which verifies that $\supp(\nu) \subset \partial G$, and that $\nu \ast |\r|^{d-2} = \mu \ast |\r|^{d-2}$ for all $\r \in \R^d \setminus G$ and $\nu \ast |\r|^{d-2} \leq \mu \ast |\r|^{d-2}$ for all $\r \in \overline{G}$. Moreover, we have $\intRd \nu \leq \intRd \mu$.
\end{theorem}
This concludes the proof of \Cref{thm:charge}.
\end{proof}
We now turn to the proof of \Cref{limit_vc}.
\begin{proof}[Proof of \Cref{limit_vc}]
It suffices to consider the case where $\mu$ is radial, that is $\mu(\cR A) = \mu(A)$ for all Borel sets $A \subset \R^d$ and all rotations $\cR \in SO(d)$. Indeed, we always have
\begin{equation}\label{eq:radial_liminf}
0 \leq \liminf_{|\r| \to \infty} \mu \ast |\r|^{2-d} \leq \liminf_{|\r| \to \infty} \widetilde{\mu} \ast |\r|^{2-d}
\end{equation}
where $\widetilde{\mu} \in \cM_+(\R^d)$ is the radial measure defined as 
\begin{equation}
\widetilde{\mu}(A) = \int_{SO(d)} \mu(\cR A) \nu(\d\cR).
\end{equation}
where $\nu$ is the \e{Haar measure} of $SO(d)$. Therefore, let us suppose that $\mu$ is radial.
According to \e{Newton's theorem} (see, \emph{e.g.} \cite[Thm 9.7]{lieb1996analysis}), we have 
\begin{equation}\label{eq:newton}
\mu \ast |\r|^{2-d} = \frac{1}{|\r|^{d-2}}\int_{B_{|\r|}} \mu(\d \r') + \int_{\R^d \setminus B_{|\r|}} \frac{\mu(\d \r')}{|\r'|^{d-2}}.
\end{equation}
Let $(\r_n)_{n \geq 0}$ be any sequence such that $|\r_n| \to \infty$ as $n \to \infty$, and suppose that $\r_0 = 0$. We write $r_n = |\r_n|$ and $\mu(r_n) = \mu(B_{r_n})$. Up to a subsequence, since $\mu$ is locally finite, we can suppose that 
\begin{equation}\label{eq:hypothesis}
\mu(r_{n-1}) \leq \sqrt{r_n}, \quad \text{for all } n \geq 1.
\end{equation}
But, we can write
\begin{equation}
\int_{\R^d} \frac{\mu(\d \r')}{|\r'|^{d-2}}  = \sum_{n \geq 1} \int_{B_{r_n} \setminus B_{r_{n-1}}} \frac{\mu(\d \r')}{|\r'|^{d-2}} \geq \sum_{n \geq 1} \frac{1}{r_{n}} (\mu(r_n) - \mu(r_{n-1})).
\end{equation}
By assumption, $\int_{\R^d} |\r'|^{2-d} \mu(\d \r') < \infty$, so it must be that 
\begin{equation}
\frac{1}{r_{n}} (\mu(r_n) - \mu(r_{n-1})) \xrightarrow[n \to \infty]{} 0.
\end{equation}
But, according to \eqref{eq:hypothesis}, we have
\begin{equation}
\frac{1}{r_n} \mu(r_{n-1}) \leq \frac{1}{\sqrt{r_n}} \xrightarrow[n \to \infty]{} 0,
\end{equation}
so that
\begin{equation}\label{eq:key}
\frac{1}{r_n} \mu(r_n) \xrightarrow[n \to \infty]{} 0.
\end{equation}
Now, still appealing to Newton's theorem, we have
\begin{equation}
 \mu \ast |\r|^{2-d}(\r_n) = \frac{1}{r_n} \mu(r_n) + \int_{\R^d \setminus B_{r_n}}\frac{\mu(\d \r')}{|\r'|^{d-2}}  \xrightarrow[n \to \infty]{} 0,
\end{equation} 
which yields the thesis of \Cref{limit_vc}.
\end{proof}

\subsection{Proof of \Cref{thm:charge_unique}}\label{sec:thm:charge_unique}The uniqueness of the dual charge is intrinsically linked to that of the uniqueness of the Kantorovich potential — and, as such, follows from the following proposition:
\begin{proposition}\label{prop:vunique}
Let $\rho \in L^1(\R^d, \R_+)$ with $\intRd \rho = N \geq 2$, and suppose that $\rho$ has connected support $\Omega$. Let $v$ and $w$ be two Lipschitz Kantorovich potentials for \pcref{prob:MOT_Coulomb_d}. Then $v$ and $w$ are equal up to an additive constant on the support $\Omega$ of $\rho$. 
\end{proposition}
\begin{proof}[Proof of \Cref{prop:vunique}]
It follows from \cite[Thm 1.15]{santambrogio2015optimal} that, for any two almost-everywhere differentiable Kantorovich potentials $v$ and $w$, we have
\begin{align}
\nabla v(\r) = \nabla w(\r) \quad\text{for almost-all } \r \in \Omega,
\end{align}
where the gradients $\nabla v$ and $\nabla w$ are to be understood in the \e{Fréchet} sense. Since Lipschitz functions are differentiable almost-everywhere according to \e{Rademacher's theorem}, the above equality makes sense. Furthermore, it is known \cite{evans1998partial} that the gradient of a Lipschitz function identifies with its \e{distributional} gradient. Therefore, it must be that
\begin{align}\label{eee}
v(\r) = w(\r) + c, \quad\text{for all }\r \in \Omega, \quad c \in \R.
\end{align}
\end{proof}
\begin{proof}[Proof of \Cref{thm:charge_unique}]
This follows immediately from \Cref{prop:vunique} by considering the distributional Laplacian in \eqref{eee}.
\end{proof}
\subsection{On duality theory at positive temperature}\label{duality_theory_temp}
We recall important facts regarding the duality theory for \pcref{prob:MOT_Coulomb_temp}. As mentioned above, the SCE problem at positive temperature can be viewed as the Legendre transform 
\begin{align}\label{eq:dual} 
F_{\SCE, \beta}(\rho) = \sup_{v} \left\{\cF_\beta(v) + \intRd v(\r) \rho(\r) \d \r \right\},
\end{align}
where $\cF_\beta(v)$ is the Helmholtz free energy of the canonical ensemble with external potential (minus) $v$, that is $$\cF_\beta(v) = -\beta^{-1} \ln z_{\beta}(v)
$$ where $z_\beta(v)$ is defined as the volume $\int_{\R^{dN}} \d G_{\beta}(v)$ of the configurational canonical ensemble $G_{\beta}(v)$ defined as the Gibbs-Boltzmann measure whose density with respect to $\mu^{\otimes N}$ is given by
\begin{equation}\label{GBm}
\d G_\beta(v)(\r_1, \dots, \r_N) = \exp\left[\textstyle- \beta\left( c(\r_1, \dots, \r_N)-\sum_{i=1}^N v(\r_i)\right)\right] \d \mu^{\otimes N}.
\end{equation}

Appealing to the strict concavity of \pcref{prob:MOT_Coulomb_temp_d}, one can formally take the functional derivative of the objective and solve for the first-order optimality condition. One finds 
\begin{align}\label{opt_cond}
\frac{\delta}{\delta v(\r)} \left( \cF_\beta(v) + \intRd v(\r) \rho(\r) \d \r \right) = \rho_{G_{\beta}(v)}(\r) - \rho(\r)
\end{align}
where $\rho_{G_{\beta}(v)}$ is the one-particle density of the ensemble $G_{\beta}(v)$, that is
\begin{align}\label{1_particle_density_ce}
\rho_{G_{\beta}(v)}(\r) = N \int_{\R^{d(N-1)}} {G_{\beta}(v)}(\r, \d\r_2, \dots, \d\r_N).
\end{align}
Therefore \pcref{prob:MOT_Coulomb_temp_d} amounts to finding the unique (up to an additive constant) potential $v_\beta$ such that the one-particle density of the associated canonical ensemble is $\rho$. We will always write $G_\beta$ for $G_\beta(v_\beta)$.

\begin{convention}\label{conv:1}
Up to an additive constant, we will always assume that $v_\beta$ verifies $z_\beta(v_\beta) = N$. 
\end{convention} 
Note that the first-order optimality condition $\rho_{G_\beta} = \rho$ rewrites under the above convention as the fixed-point equation
\begin{equation}\label{eq:Vbetaeq}
v_\beta(\r) = -\beta^{-1}\ln \int_{\R^{d(N-1)}} G_{\beta}^\r(\r_2, \dots, \r_N) \d\mu^{\otimes (N-1)}(\r_2, \dots, \r_N)
\end{equation}
for almost all $\r \in \Omega$, where $\Omega$ is the support of $\rho$ and where $G_\beta^\r$ is defined as
\begin{equation}\label{def:Gbr}
G_\beta^\r(\r_2, \dots, \r_N) = \exp\left[\textstyle- \beta\left( c(\r, \r_2, \dots, \r_N)-\sum_{i=2}^N v_\beta(\r_i)\right)\right].
\end{equation} 

\begin{convention}\label{conv:2}
Since in \eqref{GBm} the Gibbs measure $G_\beta$ is defined only with respect to $\rho$, in the first-order optimality condition \eqref{opt_cond} one can free modified $v_\beta$ outside of the support $\Omega$ of $\rho$. Furthermore, although \Cref{eq:Vbetaeq} is only valid inside of $\Omega$, the right-hand side is well-defined for all $\r$ and can be used to extend $v_\beta$ over the entire space. In what follows we will always suppose that $v_\beta$ is defined everywhere according to \Cref{eq:Vbetaeq}.
\end{convention} 

Finally, let us indicate that the equality at \eqref{eq:dual} between \pcref{prob:MOT_Coulomb_temp} and  \pcref{prob:MOT_Coulomb_temp_d} follows from classical Convex Optimization theorem, see \emph{e.g.} \cite[Thm 5.17]{simon2011convexity}, and that the above argument leading to the equation \eqref{eq:Vbetaeq} was made rigorous in \cite{chayes1984inverse} under the assumption that (\e{i.e.} see \eqref{A1})
\begin{align}
 \iint_{\R^d \times \R^d} \frac{\rho(\d \r) \rho(\d \r')}{|\r - \r'|^{d-2}} < \infty.
\end{align}
\subsection{Proof of \Cref{thm:charge_temp}}\label{sec:thm:charge_temp} Let us start with the following lemma:
\begin{lemma}\label{lem:Vsmooth}
Let $\rho \in L^1(\R^d, \R_+)$ with $\intRd \rho = N \geq 2$ be such that the assumptions \eqref{A1} and \eqref{A2} are verified. Then, the unique maximizer $v_\beta$ of \pcref{prob:MOT_Coulomb_temp_d} under \Crefrange{conv:1}{conv:2} is a twice continuously differentiable superharmonic function, \emph{i.e.} $v_\beta \in C^2(\R^d)$ and $-\Delta v_\beta \geq 0$. \end{lemma}
\begin{proof}[Proof of \Cref{lem:Vsmooth}] It follows from dominated convergence that $v_\beta$ is twice continuously differentiable. To prove that $v_\beta$ is superharmonic, it then suffices to check that $-\Delta v_\beta \geq 0$ \cite[p.56-57]{land}. We denote by $\overline{G_\beta^\r}$ the Gibbs-Boltzmann measure $G_\beta^\r$ normalized to unity. We have
\begin{align}
\nabla v_\beta(\r) = - \int_{\R^{d(N-1)}} \nabla_{\r} c(\r, \r_2, \dots, \r_N) \d \overline{G_\beta^\r}(\r_2, \dots, \r_N) 
\end{align} 
and, taking the divergence with respect to $\r$ above, we obtain
\begin{multline}
-\beta^{-1}\Delta v_\beta(\r) =  \int_{\R^{d(N-1)}} \left|\nabla_{\r} c(\r, \r_2, \dots, \r_N)\right|^2  \d \overline{G_\beta^\r} \\ - \left|\int_{\R^{d(N-1)}} \nabla_{\r} c(\r, \r_2, \dots, \r_N)  \d \overline{G_\beta^\r}\right|^2.
\end{multline}
The thesis is then obtained by appealing to Jensen's inequality. Note that for any $\r$, the function $(\r_2, \dots, \r_N) \mapsto \nabla_\r c(\r_2, \dots, \r_N)$ is integrable against the Gibbs measure $\overline{G_\beta^\r}$ for the latter vanishes 
exponentially fast as $\r_j \to \r$ for all $j = 2, \dots, N$. Therefore, the above equation makes sense for all $\r$.
\end{proof}

\begin{lemma}\label{lem:Vbounds}
Let $\rho \in L^1(\R^d, \R_+)$ with $\intRd \rho = N \geq 2$ be such that the assumptions \eqref{A1} and \eqref{A2} are verified. Then, the unique maximizer $v_\beta$ of \pcref{prob:MOT_Coulomb_temp_d} under \Crefrange{conv:1}{conv:2} is bounded in $L^\infty(\R^d)$ independently in $\beta$ in the limit $\beta \to \infty$. More precisely, we have
\begin{equation}\label{eq:Vboundup}
\sup_{\r \in \R^d} v_\beta (\r) \leq \Vert \rho\ast|\r|^{2-d} \Vert_{L^\infty}+ \frac{1}{2}D(\rho).
\end{equation}
and
\begin{equation}\label{eq:Vbounddown}
\inf_{\r \in \R^d}v_\beta(\r) \geq - (N-1) (\Vert \rho\ast|\r|^{2-d}  \Vert_{L^\infty} + \frac{1}{2}D(\rho)).
\end{equation}
\end{lemma}

\begin{proof}[Proof of \Cref{lem:Vbounds}]
The upper bound \eqref{eq:Vboundup} follows by appealing to Jensen's inequality in \Cref{eq:Vbetaeq}. Indeed, we have 
\begin{equation}
v_\beta(\r) \leq \int_{\R^{d(N-1)}} \left(c(\r, \r_2, \dots, \r_N) - \sum_{i =2}^N v_\beta(\r_i)\right) \d \mu^{\otimes N-1},
\end{equation}
which implies 
\begin{multline}
v_\beta(\r) \leq (N-1) \mu \ast |\r|^{2-d} + \binom{N-1}{2}D(\mu) \\ - (N-1)\intRd v_\beta(\r) \mu(\d \r).
\end{multline}
Now, recalling that $\mu = \rho/N$, we have by definition
\begin{equation}
F_{\SCE, \beta}(\rho) = -\beta^{-1}\ln z_\beta(v_\beta) + N\intRd v_\beta(\r) \mu(\d \r).
\end{equation}
Therefore, using the above equality and the fact that $z_\beta(v_\beta) = N$ (see \Cref{conv:1}), we obtain
\begin{multline}
v_\beta(\r) \leq (N-1) \mu \ast |\r|^{2-d} + \binom{N-1}{2}D(\mu) \\ - \frac{(N-1)}{N}(F_{\SCE, \beta}(\rho) + \beta^{-1} \ln N).
\end{multline}
Finally, from the fact that $F_{\SCE, \beta}(\rho)$ is positive and converges to $F_{\SCE}(\rho)$ as $\beta \to \infty$, we obtain the upper bound \eqref{eq:Vboundup} as claimed.
Now using the positivity of the Coulomb cost, we have
\begin{align}\label{eq:Vbounddown2}
v_\beta(\r) \geq -\beta^{-1}(N-1)\ln\left( \int_{\R^d} e^{\beta v_\beta(\r')} \d \mu(\r')\right)
\end{align}
and plugging the upper bound \eqref{eq:Vboundup} into \eqref{eq:Vbounddown2} we obtain the lower bound \eqref{eq:Vbounddown}.
\end{proof}

We can now provide the proof of  \Cref{thm:charge_temp}.
\begin{proof}[Proof of \Cref{thm:charge_temp}]
From \Cref{lem:Vsmooth} and \Cref{lem:Vbounds}, we known that $v_\beta$ (under \Crefrange{conv:1}{conv:2}) is superharmonic and bounded from below. Therefore, the proof of \Cref{thm:charge_temp}, regarding the existence of a dual charge at positive temperature, is identical to that of the proof of \Cref{thm:charge} regarding the existence of a dual charge at zero temperature.
\end{proof}

\subsection{Proofs of \Cref{thm:charge_conv} and \Cref{thm:vconv}}\label{sec:2thm}
We start with the proof of  \Cref{thm:vconv} regarding the convergence of the Kantorovich potential $v_\beta$ to a Kantorovich potential for \pcref{prob:MOT_Coulomb_d_o} in the limit $\beta \to \infty$.
\begin{lemma}\label{lem:vconvW1}
Let $\rho \in L^1(\R^d, \R_+)$ with $\intRd \rho = N \geq 2$ be such that the assumptions \eqref{A1} and \eqref{A2} are verified, and such that the support $\Omega$ of $\rho$ is bounded. Then, the unique maximizer $v_\beta$ of \pcref{prob:MOT_Coulomb_temp_d} under \Crefrange{conv:1}{conv:2} is bounded in $W^{1,\infty}(\R^d)$ independently in $\beta$ in the limit $\beta \to \infty$.
\end{lemma}
\begin{proof}[Proof of \Cref{lem:vconvW1}]
According to \Cref{lem:Vbounds}, we know that $v_\beta$ is bounded in $L^\infty(\R^d)$ uniformly in $\beta$ in the small temperature limit, so that it only remains to prove that $\nabla V_\beta$ is also bounded in $L^\infty(\R^d)$ uniformly in $\beta$ as $\beta \to \infty$. Recall that 
\begin{align}
\nabla v_\beta(\r) = - \int_{\R^{d(N-1)}} \nabla_{\r} c(\r, \r_2, \dots, \r_N) \d \overline{G_\beta^\r}(\r_2, \dots, \r_N) 
\end{align}
where $\overline{G_\beta^\r}$ is the Gibbs-Boltzmann measure $G_\beta^\r$ normalized to unity. Using the fact that $G_\beta^\r$ is symmetric, we have
\begin{align}
|\nabla v_\beta(\r)| \leq (N-1)\int_{\R^{d(N-1)}} \frac{ \d \overline{G_\beta^\r}(\r_2, \dots, \r_N)}{|\r - \r_2|^{d-1}}.
\end{align}
For all $\beta > 0$ and all $\r\in\R^d$, we have $\supp(G_\beta^\r) \subset \Omega$, where $\Omega$ is the support of $\rho$. Therefore, since $\Omega$ is by hypothesis bounded, we have
$$
\int_{\R^{d(N-1)}}\frac{ \d \overline{G_\beta^\r}(\r_2, \dots, \r_N)}{|\r - \r_2|^{d-1}} = \frac{1}{|\r|^{d-1}}+ o_{|\r| \to \infty}(1),
$$
where the $o_{|\r| \to \infty}(1)$ is independent of $\beta$. Therefore, for a large enough compact set $K$, the gradient $\nabla V_\beta$ is bounded in $L^\infty(\R^d \setminus K)$ uniformly in $\beta$ in the limit $\beta \to \infty$. It remains to control $|\nabla V_\beta|$ inside of $K$. 

Let us define the (unnormalized) Gibbs-Boltzmann measure $H_\beta^\r$ as 
$$
\d H_\beta^\r(\r_2, \dots, \r_N) = \frac{G_\beta^\r(\r_2, \dots, \r_N)}{|\r - \r_2|^{d-1}} \d \mu^{\otimes (N-1)}(\r_2, \dots, \r_N).
$$
and let $\overline{H_\beta^\r}$ be the associated probability measure, \emph{i.e.} $H_\beta^\r$ normalized to unity. We denote by $g_\beta(\r)$ the volume of $G_\beta^\r$, \emph{i.e.} $g_\beta(\r) = \int_{\R^{d(N-1)}} \d G_\beta^\r$, and by $h_\beta(\r)$ the volume of $H_\beta^\r$. We can then rewrite
\begin{align}
A(\r) := \int_{\R^{d(N-1)}}\frac{ \d \overline{G_\beta^\r}(\r_2, \dots, \r_N)}{|\r - \r_2|^{d-1}}  = \frac{h_\beta(\r)}{g_\beta(\r)}
\end{align}
According to the \e{Gibbs Variational Principle}, the free energies associated with $G_\beta^\r$ and $H_\beta^\r$ can be rewritten as 
\begin{align}
-\beta^{-1} \ln g_\beta(\r) &= F(\overline{G_\beta^\r}) = \min_{\Pro} F(\Pro) \\
\text{and } -\beta^{-1} \ln h_\beta(\r) &= F'(\overline{H_\beta^\r}) = \min_{\Pro} F'(\Pro),
\end{align}
where the minimum runs over all probability measures $\Pro$ on $\R^{d(N-1)}$ and where the functionals $F$ and $F'$ are defined as the total energies
\begin{align}
F(\Pro) &= \int_{\R^{d(N-1)}} \mathcal{H}^\r (\r_2, \dots, \r_N)\d \Pro + \beta^{-1}{\sf Ent}(\Pro | \mu^{\otimes (N-1)})\\
\text{and } \,\, F'(\Pro) &= F(\Pro) + \beta^{-1}\int_{\R^{d(N-1)}} \ln|\r - \r_2|^{d-1}\d\Pro,
\end{align}
where $\mathcal{H}^\r$ is the Hamiltonian defined as
$$
 \mathcal{H}^\r(\r_2, \dots, \r_N) = c(\r, \r_2, \dots, \r_N) - \sum_{i = 2}^N v_\beta(\r_i).
$$
Therefore, we obtain
\begin{align}
-\beta^{-1} \ln A(\r) = F'(\overline{H_\beta^\r}) - F(\overline{G_\beta^\r}) \geq F'(\overline{H_\beta^\r}) - F(\overline{H_\beta^\r}),
\end{align}
and the preceding inequality rewrites as
\begin{align}
A(\r) \leq \exp\left( - \int_{\R^{d(N-1)}} \ln |\r - \r_2|^{d-1} \d \overline{H_\beta^\r} \right).
\end{align}
Using the inequality $-\log t^{d-1} \leq t^{2-d}$ for $t > 0$, we then obtain
\begin{equation}\label{eq:Wxin3}
A(\r) \leq \exp \left( \int_{\R^{d(N-1)}} \frac{\d \overline{H_\beta^\r} }{|\r - \r_2|^{d-2}} \right).
\end{equation}
Now, let us appeal to the Gibbs Variational Principle once again to further the inequality \eqref{eq:Wxin3}. Indeed, for any probability measure $\Pro$, we have (by definition) that $F'({\overline{H_\beta^\r}}) \leq F'(\Pro)$. Using the fact that the entropy and the Coulomb cost are positive, we are led to
\begin{multline}\label{eq:Gxpb}
\int_{\R^{d(N-1)}} \frac{\d {\overline{H_\beta^\r}}}{|\r - \r_2|^{d-2}} \leq F'(\Pro) \\ + \int_{\R^{d(N-1)}} \left(\sum_{i=2}^N v_\beta(\r_i) - \beta^{-1} \log|\r - \r_N|^{d-1} \right) \d \overline{H_\beta^\r}.
\end{multline}
Using that $v_\beta$ is uniformly bounded in $\beta$ in the limit $\beta \to \infty$, say $\Vert v_\beta \Vert_{L^\infty(\R^d)} \leq M$ for all $\beta$ as $\beta \to \infty$, we further obtain
\begin{equation}\label{eq:Gxpb2}
(1 - \beta^{-1})\int_{\R^{d(N-1)}} \frac{\d {\overline{H_\beta^\r}} }{|\r - \r_2|^{d-2}} \leq F'(\Pro) + (N-1)M.
\end{equation}
Now, it suffices to put, \e{e.g.} $\Pro = |\Omega|^{-(N-1)} \mathds{1}_\Omega^{\otimes (N-1)}$, in \eqref{eq:Gxpb2}. One can check that the bound then obtained is bounded in $L^\infty(K)$ uniformly in the temperature $\beta$, yielding the thesis of \Cref{lem:vconvW1}.
\end{proof}
\begin{lemma}\label{lem:vconvAA}
Let $\rho \in L^1(\R^d, \R_+)$ with $\intRd \rho = N \geq 2$ be such that the assumptions \eqref{A1} and \eqref{A2} are verified, and let $v_\beta$ be the unique maximizer of \pcref{prob:MOT_Coulomb_temp_d} under \Crefrange{conv:1}{conv:2}. Then there exists $v_\infty \in W^{1, \infty}(\R^d)$ such that, up to a subsequence $\beta' \to \infty$, we have
\begin{equation}\label{eq:Vuniform}
v_\beta \xrightarrow[\beta \to \infty]{} v_\infty \quad \text{uniformly on every compact set}.
\end{equation}
Moreover, $v_\infty$ is a Kantorovich potential for \pcref{prob:MOT_Coulomb_d_o}.
\end{lemma}
\begin{proof}[Proof of \Cref{lem:vconvAA}]
The Kantorovich potential $v_\beta$ being bounded in $W^{1,\infty}(\R^d)$ uniformly in $\beta$ in the limit $\beta \to \infty$, it follows that there exists $v_\infty \in W^{1, \infty}(\R^d)$ such that (up to a subsequence $\beta' \to \infty$) 
\begin{equation}\label{eq:Vbetawsl}
v_\beta \to v_\infty \,\, \text{weakly in } W^{1,\infty}(\R^d) \text{ and locally uniformly in } L^\infty(\R^d).
\end{equation}
We claim that $v_\infty$ is a Kantorovich potential for \pcref{prob:MOT_Coulomb_d_o}. Since $F_{\SCE, \beta}(\rho) \to F_{\rm SCE}(\rho)$ as $\beta \to \infty$ \cite{carlier2017convergence}, we have
\begin{equation}\label{limit_temp}
F_{\SCE, \beta}(\rho) = -\beta^{-1} \ln N + \int_\Omega v_\beta \rho \xrightarrow[\beta \to \infty]{} \int_\Omega v_\infty \rho = F_{\rm SCE}(\rho).
\end{equation}
Now, by duality we have 
\begin{equation}
E_{N, \Omega}(v_\infty) + \int_\Omega v_\infty \rho \leq F_{\rm SCE}(\rho) = \int_\Omega v_\infty \rho,
\end{equation}
where the last equality follows from \eqref{limit_temp}. Therefore, we will have proved that $v_\infty$ is a Kantorovich potential for \pcref{prob:MOT_Coulomb_d_o} if we can prove that $E_{N, \Omega}(v_\infty) \geq 0$ — which will eventually imply that $E_{N, \Omega}(v_\infty) = 0$.
Let us then define
$$
A_\epsilon= \left\{ (\r_1, \dots, \r_N) \in \Omega^N: \sum_{i = 1}^N v_\infty(\r_i) > c_\alpha(\r_1, \dots, \r_N) + \epsilon \right\}.
$$
By convention, we have $z_\beta(v_\beta) = N$ for all $\beta > 0$. Appealing to \e{Fatou's lemma}, we have
\begin{equation}\label{eq:Zine}
N \geq \int_{A_\epsilon} \liminf_{\beta \to \infty} G_\beta(\r_1, \dots, \r_N) \d \mu^{\otimes N}(\r_1, \dots, \r_N)
\end{equation}
But, by definition of $A_\epsilon$, for all $(\r_1, \dots, \r_N) \in A_\epsilon$ we have 
\begin{equation}\label{eq:limiteHV}
\liminf_{\beta \to \infty} G_\beta(\r_1, \dots, \r_N) = + \infty.
\end{equation}
Therefore, in regards of \eqref{eq:Zine} and \eqref{eq:limiteHV}, it must be that
$$
\mu^{\otimes N}(A_\varepsilon) = 0.
$$
Since $A_\varepsilon$ is open in $\Omega^N$, we necessarily have that $A_\epsilon$ is empty for all $\epsilon > 0$. Therefore, we obtain $E_{N,\Omega}(v_\infty) \geq 0$ as wanted, yielding the thesis that $v_\infty$ is a Kantorovich potential for  \Cref{prob:MOT_Coulomb_d_o}.
\end{proof}
\begin{proof}[Proof of \Cref{thm:vconv}]
The proof of  \Cref{thm:vconv} now follows entirely from \Cref{lem:vconvW1} and \Cref{lem:vconvAA}.
\end{proof}

We now turn to the proof of \Cref{thm:charge_conv}. Given two measures $\mu, \nu \in \cM(\R^d)$, their \e{Coulomb energy} is defined as
$$
D(\mu, \nu) = \iint_{\R^d \times \R^d} \frac{\mu(\d\r) \nu(\d\r')}{|\r - \r'|^{d-2}},
$$
and we use the shorthand notation $D(\mu) := D(\mu, \mu)$ to denote the \e{self-energy} of $\mu$. If we define the space $\cE(\R^d) \subset \cM(\R^d)$ as 
$$
\cE(\R^d) = \left\{\mu \in \cM(\R^d) : D(\mu) < \infty \right\},
$$
then $(\mu, \nu) \mapsto D(\mu,\nu)$ defines an inner product which endows $\cE(\R^d)$ with a Hilbert space structure \cite[Thm 1.18]{land}. In particular, the weak topology on $\cE(\R^d)$ is defined as follows: given $(\mu_n)_n \subset \cE(\R^d)$, we say that the sequence $(\mu_n)_n$ weakly converge (in energy) to $\mu$ if, for all $\nu \in \cE(\R^d)$, we have 
$$
D(\mu_n, \nu) \xrightarrow[n \to \infty]{} D(\mu, \nu).
$$
This weak topology is stronger than that of the vague topology on $\cM(\R^d)$ \cite[Lem 1.3]{land}. In particular, $\{ \mu \in \cE(\R^d) : D(\mu) < c \}$ is compact for the vague topology for any $c > 0$. 
\begin{proof}[Proof of \Cref{thm:charge_conv}]
Let $\rhoextbeta$ be the dual charge associated with $v_\beta$ which is not yet ``swept'' onto $\Omega$ using \Cref{thm:bal}. We claim that $D(\rhoextbeta)$ is bounded uniformly in $\beta$ in the limit $\beta \to \infty$. Integrating by parts, we have
$$
D(\rhoextbeta) = \frac{1}{c_d} \int_{\R^d} |\nabla v_\beta(\r)|^2 \d \r \quad \text{where } c_d = \frac{d(d-2)\pi^{d/2}}{\Gamma(\frac d2 +1)}.
$$
Now, as in the proof of \Cref{thm:vconv}, we have
$$
|\nabla v_\beta(\r)|^2 \leq (N-1)^2\int_{\R^{d(N-1)}} \frac{\d \overline{G_\beta^\r}(\r_2, \dots, \r_N)}{|\r - \r_2|^{2(d-1)}}.
$$
Outside a large enough compact set $K$, the above integral is bounded in $L^2(\R^d \setminus K)$ uniformly in $\beta$ in the limit $\beta \to \infty$. It then remains to control the $L^2$-norm of $\nabla v_\beta$ inside $K$. The strategy to do so is completely analogous to that of the proof of \Cref{thm:vconv}. 

Now, if $\rhoextbeta$ is swept onto $\Omega$, its self-energy does not increase \cite[Thm. 4.4.]{land}. Therefore, according to what precedes, the self-energy of the balayage measure remains uniformly bounded in $\beta$ in the limit $\beta \to \infty$. As such, there exists $\rho_\infty \in \cE(\R^d)$ such that 
\begin{equation}\label{eq:rhobOc}
\rhoextbeta \wsl \rho_\infty \quad \text{in } \cM(\R^d),
\end{equation}
and we have that $\rho_\infty$ is a dual charge for \pcref{prob:MOT_Coulomb_d_o}. Indeed, we know that (up to a subsequence; see \Cref{thm:vconv}) 
\begin{equation}\label{eq:VbOc}
\rhoextbeta \ast |\r|^{2-d} \xrightarrow[\beta \to \infty]{} v_\infty \quad \text{uniformly on } \Omega,
\end{equation}
where $v_\infty$ is a Kantorovich potential for \pcref{prob:MOT_Coulomb_d_o}. Therefore, using \eqref{eq:rhobOc}, we obtain $\rho_\infty \ast |\r|^{2-d}= v_\infty$, yielding that $\rho_\infty$ is a dual charge for  \pcref{prob:MOT_Coulomb_d_o}.
\end{proof}
\bibliographystyle{acm}
\bibliography{mybib}

\begin{thebibliography}{10}

\bibitem{alfonsi2021approximation}
{\sc Alfonsi, A., Coyaud, R., Ehrlacher, V., and Lombardi, D.}
\newblock Approximation of optimal transport problems with marginal moments
  constraints.
\newblock {\em Mathematics of Computation 90}, 328 (2021), 689--737.

\bibitem{altschuler2020polynomial}
{\sc Altschuler, J.~M., and Boix-Adsera, E.}
\newblock Polynomial-time algorithms for multimarginal optimal transport
  problems with decomposable structure.
\newblock {\em arXiv preprint arXiv:2008.03006\/} (2020).

\bibitem{altschuler2021hardness}
{\sc Altschuler, J.~M., and Boix-Adsera, E.}
\newblock Hardness results for multimarginal optimal transport problems.
\newblock {\em Discrete Optimization 42\/} (2021), 100669.

\bibitem{benamou2016numerical}
{\sc Benamou, J.-D., Carlier, G., and Nenna, L.}
\newblock A numerical method to solve multi-marginal optimal transport problems
  with coulomb cost.

\bibitem{buttazzo2018continuity}
{\sc Buttazzo, G., Champion, T., and De~Pascale, L.}
\newblock Continuity and estimates for multimarginal optimal transportation
  problems with singular costs.
\newblock {\em Applied Mathematics \& Optimization 78}, 1.

\bibitem{buttazzo2012optimal}
{\sc Buttazzo, G., De~Pascale, L., and Gori-Giorgi, P.}
\newblock Optimal-transport formulation of electronic density-functional
  theory.
\newblock {\em Physical Review A 85}, 6 (2012).

\bibitem{cances2007theoretical}
{\sc Cances, E., Legoll, F., and Stoltz, G.}
\newblock Theoretical and numerical comparison of some sampling methods for
  molecular dynamics.
\newblock {\em ESAIM: Mathematical Modelling and Numerical Analysis 41}, 2
  (2007), 351--389.

\bibitem{carlier2017convergence}
{\sc Carlier, G., Duval, V., Peyr{\'e}, G., and Schmitzer, B.}
\newblock Convergence of entropic schemes for optimal transport and gradient
  flows.
\newblock {\em SIAM Journal on Mathematical Analysis 49}, 2 (2017), 1385--1418.

\bibitem{chayes1984inverse}
{\sc Chayes, J., Chayes, L., and Lieb, E.~H.}
\newblock The inverse problem in classical statistical mechanics.
\newblock {\em Communications in Mathematical Physics 93}, 1 (1984), 57--121.

\bibitem{colombo2015multimarginal}
{\sc Colombo, M., De~Pascale, L., and Di~Marino, S.}
\newblock Multimarginal optimal transport maps for one--dimensional repulsive
  costs.
\newblock {\em Canadian Journal of Mathematics 67}, 2 (2015).

\bibitem{colombo2019continuity}
{\sc Colombo, M., Di~Marino, S., and Stra, F.}
\newblock Continuity of multimarginal optimal transport with repulsive cost.
\newblock {\em SIAM Journal on Mathematical Analysis 51}, 4 (2019), 2903--2926.

\bibitem{cotar2013density}
{\sc Cotar, C., Friesecke, G., and Kl{\"u}ppelberg, C.}
\newblock Density functional theory and optimal transportation with coulomb
  cost.
\newblock {\em Communications on Pure and Applied Mathematics 66}, 4 (2013).

\bibitem{cuturi2013sinkhorn}
{\sc Cuturi, M.}
\newblock Sinkhorn distances: Lightspeed computation of optimal transport.
\newblock In {\em Advances in neural information processing systems\/} (2013),
  pp.~2292--2300.

\bibitem{evans1998partial}
{\sc Evans, L.~C.}
\newblock {\em Partial differential equations}, vol.~19.
\newblock Rhode Island, USA, 1998.

\bibitem{evans1992density}
{\sc Evans, R.}
\newblock Density functionals in the theory of nonuniform fluids.
\newblock {\em Fundamentals of inhomogeneous fluids 1\/} (1992), 85--176.

\bibitem{friesecke2022strong}
{\sc Friesecke, G., Gerolin, A., and Gori-Giorgi, P.}
\newblock The strong-interaction limit of density functional theory.
\newblock {\em arXiv preprint arXiv:2202.09760\/} (2022).

\bibitem{friesecke2022genetic}
{\sc Friesecke, G., Schulz, A.~S., and Vogler, D.}
\newblock Genetic column generation: Fast computation of high-dimensional
  multimarginal optimal transport problems.
\newblock {\em SIAM Journal on Scientific Computing 44}, 3 (2022),
  A1632--A1654.

\bibitem{gerolin2020multi}
{\sc Gerolin, A., Kausamo, A., and Rajala, T.}
\newblock Multi-marginal entropy-transport with repulsive cost.
\newblock {\em Calculus of Variations and Partial Differential Equations 59\/}
  (2020), 1--20.

\bibitem{kohn1965self}
{\sc Kohn, W., and Sham, L.~J.}
\newblock Self-consistent equations including exchange and correlation effects.
\newblock {\em Physical review 140}, 4A (1965), A1133.

\bibitem{land}
{\sc Landkof, N.~S.}
\newblock {\em Foundations of modern potential theory}.

\bibitem{leonard2012schrodinger}
{\sc L{\'e}onard, C.}
\newblock From the schr{\"o}dinger problem to the monge--kantorovich problem.
\newblock {\em Journal of Functional Analysis 262}, 4 (2012), 1879--1920.

\bibitem{levy1979universal}
{\sc Levy, M.}
\newblock Universal variational functionals of electron densities, first-order
  density matrices, and natural spin-orbitals and solution of the
  v-representability problem.
\newblock {\em Proceedings of the National Academy of Sciences 76}, 12 (1979),
  6062--6065.

\bibitem{lewin2018semi}
{\sc Lewin, M.}
\newblock Semi-classical limit of the levy--lieb functional in density
  functional theory.
\newblock {\em Comptes Rendus Math{\'e}matiques 356}, 4 (2018).

\bibitem{lewin2018statistical}
{\sc Lewin, M., Lieb, E.~H., and Seiringer, R.}
\newblock Statistical mechanics of the uniform electron gas.
\newblock {\em Journal de l'{\'E}cole polytechnique—Math{\'e}matiques 5\/}
  (2018), 79--116.

\bibitem{lewin2019universal}
{\sc Lewin, M., Lieb, E.~H., and Seiringer, R.}
\newblock Universal functionals in density functional theory.
\newblock {\em arXiv preprint arXiv:1912.10424\/} (2019).

\bibitem{lieb83dft}
{\sc Lieb, E.~H.}
\newblock Density functionals for coulomb systems.
\newblock {\em International Journal of Quantum Chemistry 24}, 3 (1983).

\bibitem{lieb1996analysis}
{\sc Lieb, E.~H., and Loss, M.}
\newblock Analysis, graduate stud. math., vol. 14.
\newblock In {\em Amer. Math. Soc\/} (1996).

\bibitem{mendl2013kantorovich}
{\sc Mendl, C.~B., and Lin, L.}
\newblock Kantorovich dual solution for strictly correlated electrons in atoms
  and molecules.
\newblock {\em Physical Review B 87}, 12 (2013).

\bibitem{nenna2016numerical}
{\sc Nenna, L.}
\newblock {\em Numerical methods for multi-marginal optimal transportation}.
\newblock PhD thesis, 2016.

\bibitem{nutz2021entropic}
{\sc Nutz, M., and Wiesel, J.}
\newblock Entropic optimal transport: Convergence of potentials.
\newblock {\em Probability Theory and Related Fields\/} (2021), 1--24.

\bibitem{pass2014multimarginal}
{\sc Pass, B.}
\newblock Multi-marginal optimal transport: theory and applications.
\newblock {\em ESAIM: Mathematical Modelling and Numerical Analysis 49}, 6
  (2015), 1771--1790.

\bibitem{peyre2019computational}
{\sc Peyr{\'e}, G., and Cuturi, M.}
\newblock Computational optimal transport.
\newblock {\em Foundations and Trends in Machine Learning 11}, 5-6 (2019).

\bibitem{rasanen2011strictly}
{\sc R{\"a}s{\"a}nen, E., Seidl, M., and Gori-Giorgi, P.}
\newblock Strictly correlated uniform electron droplets.
\newblock {\em Physical Review B 83}, 19 (2011), 195111.

\bibitem{roberts1996exponential}
{\sc Roberts, G.~O., and Tweedie, R.~L.}
\newblock Exponential convergence of langevin distributions and their discrete
  approximations.
\newblock {\em Bernoulli\/} (1996), 341--363.

\bibitem{rousset2010free}
{\sc Rousset, M., Stoltz, G., and Lelievre, T.}
\newblock {\em Free energy computations: a mathematical perspective}.
\newblock World Scientific, 2010.

\bibitem{santambrogio2015optimal}
{\sc Santambrogio, F.}
\newblock Optimal transport for applied mathematicians.
\newblock {\em Birk{\"a}user, NY 55}, 58-63 (2015).

\bibitem{seidl1999strong}
{\sc Seidl, M.}
\newblock Strong-interaction limit of density-functional theory.
\newblock {\em Physical Review A 60}, 6 (1999), 4387.

\bibitem{seidl1999strictly}
{\sc Seidl, M., Perdew, J.~P., and Levy, M.}
\newblock Strictly correlated electrons in density-functional theory.
\newblock {\em Physical Review A 59}, 1 (1999), 51.

\bibitem{simon2011convexity}
{\sc Simon, B.}
\newblock {\em Convexity: an analytic viewpoint}, vol.~187.
\newblock Cambridge University Press, 2011.

\bibitem{singh1991density}
{\sc Singh, Y.}
\newblock Density-functional theory of freezing and properties of the ordered
  phase.
\newblock {\em Physics Reports 207}, 6 (1991), 351--444.

\bibitem{wu2008density}
{\sc Wu, J.}
\newblock Density functional theory for liquid structure and thermodynamics.

\bibitem{yang1976molecular}
{\sc Yang, A.~J., Fleming~III, P.~D., and Gibbs, J.~H.}
\newblock Molecular theory of surface tension.
\newblock {\em The Journal of Chemical Physics 64}, 9 (1976), 3732--3747.

\end{thebibliography}
\end{document}